\documentclass[11pt,twoside]{amsart}
\usepackage{latexsym,amssymb,amsmath}
\usepackage{mathtools}
\usepackage[all]{xy}

\numberwithin{equation}{section}
\newtheorem{theorem}{Theorem}[section]
\newtheorem{lemma}[theorem]{Lemma}
\newtheorem{proposition}[theorem]{Proposition}
\newtheorem{corollary}[theorem]{Corollary}

\theoremstyle{definition}
\newtheorem{definition}[theorem]{Definition} 
\newtheorem{procedure}[theorem]{Procedure} 
\newtheorem{remark}[theorem]{Remark}
\newtheorem{example}[theorem]{Example}

\begin{document}


\title{Linear codes over signed graphs}

\thanks{The first and third authors were supported by SNI, Mexico.
Second author was 
supported by a scholarship from CONACYT, Mexico}

\author[J. Mart\'\i nez-Bernal]{Jos\'e Mart\'\i nez-Bernal}
\address{
Departamento de
Matem\'aticas\\
Centro de Investigaci\'on y de Estudios Avanzados del IPN\\
Apartado Postal
14--740 \\
07000 Mexico City, Mexico.
}
\email{jmb@math.cinvestav.mx}

\author[M. A. Valencia-Bucio]{Miguel A. Valencia-Bucio}
\address{
Departamento de
Matem\'aticas\\
Centro de Investigaci\'on y de Estudios
Avanzados del
IPN\\
Apartado Postal
14--740 \\
07000 Mexico City, Mexico
}
\email{mavalencia@math.cinvestav.mx}

\author[R. H. Villarreal]{Rafael H. Villarreal}
\address{
Departamento de
Matem\'aticas\\
Centro de Investigaci\'on y de Estudios
Avanzados del
IPN\\
Apartado Postal
14--740 \\
07000 Mexico City, Mexico
}
\email{vila@math.cinvestav.mx}

\begin{abstract} 
We give formulas, in terms of graph theoretical invariants, for the
minimum distance and the generalized Hamming weights of
the linear code generated by the rows of the incidence matrix of 
a signed graph over a finite field, and for those of its dual code.
Then we determine the regularity of the ideals of circuits 
and cocircuits of a signed graph, and prove an algebraic
formula in terms of the multiplicity for the frustration index of an
unbalanced signed graph.    
\end{abstract}

\maketitle 

\section{Introduction}\label{intro-section}

The generalized Hamming weights (GHWs) of a linear code are parameters
of interest in many applications 
\cite{rth-footprint,Pellikaan,Johnsen,olaya,schaathun-willems,tsfasman,Wei,wei-yang,Yang}
and they have been nicely related to  
the graded Betti numbers of the ideal of cocircuits of the matroid
of a linear code \cite{JohVer,Johnsen}, to the nullity
function of the dual matroid of a linear code \cite{Wei}, and to the
enumerative combinatorics of linear codes 
\cite{Britz,Johnsen-Roksvold-Verdure,Klove-1992,MacWilliams-Sloane}.
Because of this, 
their study has attracted considerable attention, 
but determining them is in general a
difficult problem. The notion of generalized Hamming weight was
introduced by Helleseth, Kl{\o}ve and 
Mykkeltveit in \cite{HKM} and was first used systematically by Wei in
\cite{Wei}. 
For convenience we recall this notion. Let $K=\mathbb{F}_q$ be a finite
field and let $C$ be an  
$[m,k]$-\textit{linear
code} of {\it length} $m$ and {\it dimension} $k$, 
that is, $C$ is a linear subspace of $K^m$ with $k=\dim_K(C)$. Let
$1\leq r\leq k$ be an integer.   
Given a linear subspace $D$ of $C$, the \textit{support} of
$D$, denoted $\chi(D)$, is the set of nonzero positions of $D$, that is,
$$
\chi(D):=\{i\,\colon\, \exists\, (a_1,\ldots,a_m)\in D,\, a_i\neq 0\}.
$$
\quad The $r$-th \textit{generalized Hamming weight} of $C$, denoted
$\delta_r(C)$, is given by
$$
\delta_r(C):=\min \{|\chi(D)|\colon D \mbox{ is a subspace of }C,\, \dim_K (D)=r\}.
$$
\quad 
As usual we call the set $\{\delta_1(C), \ldots, \delta_k(C)\}$ the 
\textit{weight hierarchy} of the linear code $C$. 
The $1$st Hamming weight of $C$ is the \textit{minimum
distance} $\delta(C)$ of $C$, that is, one has
$$
\delta_1(C)=\delta(C)=\min\{\omega(\mathbf{x})\colon \mathbf{x} \in C \setminus
\{0\}\},
$$ 
where $\omega(\mathbf{x})$ is the Hamming weight of the vector $\mathbf{x}$,
i.e., the number of non-zero entries of $\mathbf{x}$. To determine the
 minimum distance
is essential to find good error-correcting codes
\cite{MacWilliams-Sloane}. 

The notion of generalized Hamming weights for linear codes was extended to
matroids by Britz, Johnsen, Mayhew and 
Shiromoto \cite[p.~332]{BJMS} as we now explain. 

Let $M$ be a matroid with ground set $E$, rank function $\rho$,
nullity function $\eta$, and let $M^*$ be its dual matroid. 
The $r$-th \textit{generalized Hamming weight} of $M$, denoted
$d_r(M)$, is given by
\begin{equation*}
d_r(M):=\min\{|X|\colon X \subseteq E\mbox{ and }\eta(X)=r\}\ \mbox{ for
}\ 1\leq r\leq \eta(E).
\end{equation*}
\quad A major result of Johnsen and Verdure \cite{JohVer} shows 
that the GHWs of a matroid can be read off the minimal graded free
resolution of the Stanley--Reisner ideal of the independence
complex of the matroid \cite[Theorem~2]{JohVer} (see Theorem~\ref{JV3}).

We can associate to an $[m,k]$-linear code $C$  
the vector matroid $M[A]$ on the ground set $E=\{1,\ldots,m\}$, where $A$
is a generator matrix of $C$. The rank function (resp. nullity) of
$M[A]$ 
is given by 
$\rho(X)=\text{rank}(A_X)$ (resp. $\eta(X)=|X|-\rho(X)$) for $X\subseteq
E$, where $A_X$ is the submatrix of $A$ obtained by picking the
columns indexed by $X$. It
can be verified that the matroid $M[A]$ does not depend on the
generator matrix we choose. We call $M[A]$ the \textit{matroid} of
$C$. If $H$ is a parity check matrix of $C$, then $M[A]^*=M[H]$ and
$M[A]=M[H]^*$.  
By Lemma~\ref{d-delta-relations}, one has
$$
\delta_r(C^\perp)=d_r(M[A])\ \mbox{ for }\ 1\leq r\leq m-k\mbox{ and }
\delta_r(C)=d_r(M[A]^*)\mbox{ for }1\leq r\leq k. 
$$
\quad Thus computing GHWs of vector matroids is equivalent to 
computing those of linear codes.  
This relationship between the GHWs of linear codes and those of
vector matroids is attributed to Wei \cite[Theorem~2]{Wei} (cf.
Theorem~\ref{Wei-theorem}). 
In this work we study GHWs of linear codes defined over
signed graphs, combining the theory of 
GHWs of matroids \cite{Britz-bond-cycle,BJMS,JohVer,Johnsen} and the
combinatorial 
structure of signed-graphic matroids
\cite{Zaslavsky-signed-graphs,biased-graphs-I,Zaslavsky-biased-graphs-II}
that we introduce next.  

A \textit{signed graph} $G_\sigma$ is a pair
$(G,\sigma)$ consisting of a multigraph $G$ with vertex set $V(G)$ and
edge set $E(G)$ (loops and multiple edges
are permitted), and a mapping $\sigma\colon E(G)\rightarrow\{\pm\}$,
that assigns a 
sign to each edge. If no loops or multiple edges are permitted, $G$ is
called a \textit{simple graph} and $G_\sigma$ is called a
\textit{signed simple graph}. In particular, the signed
graph with $\sigma(e)=+$ (resp. $\sigma(e)=-$) for all $e$, denoted 
$G_+$ (resp. $G_{-}$), is called the \textit{positive signed graph} (resp.
\textit{negative signed graph}) on $G$. There are 
more general definitions of signed graphs, where the edge set
includes empty loops and half 
edges, that are essential to represent root systems
\cite{Zaslavsky-signed-graphs}. 

Let $G_\sigma$ be a signed graph. A \textit{cycle} of $G_\sigma$ is a
simple closed 
path in $G$. A cycle with an even number of negative edges is called
\textit{balanced}.  
A signed graph is \textit{balanced} if every cycle is balanced.  
An isolated vertex is regarded as balanced. A \textit{bowtie} 
of $G_\sigma$ is the union of two unbalanced cycles which meet at a single
vertex or the union of two vertex-disjoint unbalanced cycles and a
simple path which meets one cycle at each end and is otherwise
disjoint from them.

A central result of Zaslavsky
\cite[Theorem~5.1]{Zaslavsky-signed-graphs} shows the existence of a matroid
$M(G_\sigma)$ with ground set $E(G)$, called the
\textit{signed-graphic matroid} of $G_\sigma$, 
whose rank function is
\[\rho(X)=|V(G)|-c_0(X)\mbox{ for }X\subseteq E(G),
\]
where $c_0(X)$ is the number of balanced connected components of the
signed subgraph with edge set $X$ and vertex set $V(G)$. The circuits
of $M(G_\sigma)$ are the balanced cycles and the bowties of
$G_\sigma$. The
circuits of $M(G_\sigma)$ are called the
\textit{circuits} of $G_\sigma$. 

If $G_\sigma=G_+$, the signed-graphic matroid $M(G_+)$ is the \textit{graphic
matroid} $M(G)$ of $G$ whose circuits are the cycles of $G$
\cite{oxley,welsh}. If $G_\sigma=G_-$, the signed-graphic matroid
$M(G_-)$ is the \textit{even cycle} matroid
\cite{Zaslavsky-signed-graphs} whose circuits are the even
cycles and the bowties of $G_-$. The
circuits of the matroids $M(G_+)$, $M(G_-)$ and 
those of their dual matroids---as well as the related notion of
an elementary integral vector---occur in coding theory 
\cite{Dankelmann-Key-Rodrigues,sole-zaslavsky}, convex analysis
\cite{Rock}, 
the theory of toric ideals of graphs 
\cite{bermejo-gimenez-simis,graphs,circuitos,bowtie,Vi3,handbook}, 
and in matroid theory \cite{oxley,Simoes-Pereira,Zaslavsky-signed-graphs,Zaslavsky}.

The content of this paper is as follows. In
Section~\ref{matroids-codes} we briefly introduce matroids and 
present some well known results about GHWs of matroids and linear
codes.

In what follows $G_\sigma$ denotes a signed graph with $s$ 
vertices, $m$ edges, $c$ connected components, 
and $c_0$ balanced components, and $K$ denotes a finite field
$\mathbb{F}_q$ of 
characteristic
$p$. The \textit{incidence matrix code} of
$G_\sigma$ over the field $K$, 
denoted by $C$, is the linear code generated by 
the row vectors of the incidence matrix of
$G_\sigma$ (Definition~\ref{incidence-matrix}). In
Section~\ref{GHW-section} we present our main results on the
generalized Hamming weights 
of incidence matrix codes of signed graphs and those of their dual
codes, and describe the GHWs of the  
signed-graphic matroid of a signed graph and those of its dual
matroid, in terms of the combinatorics of the signed graph.

The \textit{frustration index} of $G_\sigma$, denoted $\varphi(G_\sigma)$, is the
smallest number of edges whose deletion from $G_\sigma$ leaves a
balanced signed graph. The minimum distance of $C$ is
bounded from above by $\varphi(G_\sigma)$ if $p\neq 2$ (Remark~\ref{dec8-18}). 
We are interested in the following related
invariant.   
The $r$-th \textit{cogirth} of $G_\sigma$, denoted
$\upsilon_r(G_\sigma)$, is the
minimum number of edges whose removal results in a signed graph with $r$ 
balanced components. If $r=1$ and $G_\sigma$ is connected, $\upsilon_1(G_\sigma)$ is the 
\textit{cogirth} of $M(G_\sigma)$, that is, the minimum size of a
cocircuit of $M(G_\sigma)$ (Lemma~\ref{jan15-19}). We denote 
$\upsilon_1(G_\sigma)$ simply by $\upsilon(G_\sigma)$. 
The $r$-th \textit{edge connectivity} of $G_\sigma$, denoted
$\lambda_r(G_\sigma)$ or $\lambda_r(G)$, 
is the minimum number of edges whose removal results in a signed graph with
$r+1$ connected components. Note that the $r$-th edge connectivity is a property
of the underlying multigraph $G$, that is, it is independent of
$\sigma$. If $r=1$, $\lambda_1(G_\sigma)$ is the
edge connectivity of $G_\sigma$ and is denoted by
$\lambda(G_\sigma)$. 
We will relate these graph invariants to the generalized Hamming
weights and the minimum distance of incidence matrix codes.

Our main results on linear codes are the following. First, we give
graph theoretical formulas for the generalized Hamming weights of the incidence matrix code
 of a signed graph.

\noindent {\bf Theorem~\ref{pepe-vila-2019}}\textit{\ If $C$ is the
incidence matrix code of a connected signed graph $G_\sigma$, 
then 
$$ 
\delta_r(C)=\begin{cases}
\upsilon_r(G_\sigma)&\text{if } p\neq 2,\, G_\sigma\textit{ is
unbalanced} \textit{ and }1\leq r\leq s,\\
\lambda_r(G)&\text{if }p=2 \textit{ and } 1\leq r\leq s-1,\\
\lambda_r(G)&\text{if }G_\sigma\textit{ is balanced} \textit{ and }1\leq r\leq s-1.
\end{cases}
$$
}
\quad We show that the formulas of  \cite[Corollary~2.13]{linear-codes} for
the generalized Hamming weights of incidence matrix codes of simple
graphs can be extended to multigraphs
(Corollary~\ref{pepe-vila-2018-coro1}). 
Then we show combinatorial formulas for the minimum
distance of the incidence matrix code of a signed graph 
\cite[Proposition~9.2.4]{oxley} (Corollary~\ref{Dankelmann-etal}). 

A family of circuits $\{C_i\}_{i=1}^r$ of a matroid $M$ is called
\textit{non-redundant} if $\bigcup_{i\neq j}C_i\subsetneq\bigcup_{i=1}^rC_i$
for $j=1,\ldots,r$ \cite{Britz-bond-cycle,Brylawski}. Our next result gives graph 
theoretical formulas for the generalized Hamming weights of the
dual code of the incidence matrix code of a signed graph. Part (a)
extends the analogous result for graphs of Britz
\cite[Section~3]{Britz-bond-cycle}.

\noindent {\bf Theorem~\ref{pepe-vila-2019-dual}}\textit{\ 
Let $C$ be the
incidence matrix code of a connected signed graph $G_\sigma$.
\begin{itemize}
\item[(a)] If $p=2$ or $G_\sigma$ is balanced, 
and $1\leq r\leq m-s+1$ $($resp. $1\leq r\leq s-1$$)$, then
$\delta_r(C^\perp)$ $($resp. $\delta_r(C)$$)$ is the minimum number
of edges of $G$ forming a union of $r$
non-redundant cycles $($resp. cocycles$)$ of $G$. 
\item[(b)] If $p\neq 2$ and $1\leq r\leq m-s$ $($resp. $1\leq r\leq s$$)$, then
$\delta_r(C^\perp)$ $($resp. $\delta_r(C)$$)$ is the 
minimum number of edges of $G$ forming a union of $r$
non-redundant balanced cycles and bowties $($resp. cocircuits$)$ of $G_\sigma$. 
\end{itemize}
}
\quad If $C$ is
the incidence matrix code of a connected digraph $\mathcal{D}$ 
and $G$ is its underlying multigraph, we show that 
$\delta_r(C)=\lambda_r(G)$ and give graph
theoretical formulas for the generalized Hamming weights of the dual code $C^\perp$
(Corollary~\ref{pepe-vila-2018-digraph}). For a connected multigraph,
we give formulas for the GHWs of the dual of its incidence matrix 
code (Corollary~\ref{feb4-19}).

The main result of Section~\ref{regularity-section} gives explicit formulas 
for the regularity of the ideals of circuits and
cocircuits 
of the vector matroid of the incidence matrix of a
signed graph (Theorem~\ref{regularity-ideal}). 
This invariant is a measure for the complexity of the minimal graded free
resolution of these ideals and has been used to study polynomial
interpolation problems \cite{eisenbud-syzygies}.

Let $M$ be the matroid of $C$. By 
Theorems~\ref{pepe-vila-2019} and \ref{pepe-vila-2019-dual}, 
one has graph theoretical formulas for the weight hierarchies of $C$ and
$C^\perp$. On the other hand, 
using \textit{Macaulay}$2$ \cite{mac2}, the package \textit{Matroids}
\cite{matroids}, and the formulas of Johnsen and Verdure
(Theorem~\ref{JV3}, Corollary~\ref{dec30-18}), we can
compute the weight hierarchies of $C$ and $C^\perp$. 
Hence, our results can be used to compute the 
$r$-th \textit{cogirth} $\upsilon_r(G_\sigma)$ of
$G_\sigma$ and the $r$-th \textit{edge connectivity}
$\lambda_r(G_\sigma)$ of $G_\sigma$. The main result of
Section~\ref{frustration-index-section} is an algebraic formulation
for the frustration index of $G_\sigma$---in terms of the degree or multiplicity 
of graded ideals---that can be used
to compute or estimate this number using \textit{Macaulay}$2$
\cite{mac2} (Theorem~\ref{vila-zaslavsky}, Example~\ref{example-graph4}). If $G$ is a graph, 
the frustration index of $G_-$ is the \textit{edge biparticity} of $G$, that is, the
minimum number of edges whose removal makes the graph bipartite.
In Section~\ref{examples-section} we illustrate how to use our
results in practice 
with some examples. 

Our main results and their proofs show that the weight hierarchies of the incidence
matrix code $C$ and its dual code $C^\perp$ of a signed graph
$G_\sigma$ can be computed using
the field $\mathbb{Q}$ of rational numbers as the ground field.
To compute the GHWs of $C$ and $C^\perp$ over a finite field
$\mathbb{F}_q$ of characteristic $p$, we use the incidence matrix of
$G_\sigma$ (resp. $G_+$) over the field $\mathbb{Q}$ if $p\neq 2$
(resp. $p=2$). One
can also use the rational numbers to compute the cycles, circuits, and 
cocircuits of a signed graph, as well as its $r$-th cogirth, 
frustration index, and $r$-th edge connectivity. In
Appendix~\ref{procedures-section} we give 
procedures for \textit{Macaulay}$2$ \cite{mac2} that allow us to 
obtain this information for graphs with a small number
of vertices, see the examples of Section~\ref{examples-section}. 
The package
\textit{Matroids} \cite{matroids} plays an important role here 
because it computes the circuits and cocircuits of vector matroids
over the field of rational numbers; however 
the problem of computing all circuits of a
vector matroid is likely to be NP-hard \cite{complexity,Vardy} (cf.
\cite[p.~76]{JohVer}). 
The minimum distance of any linear
code can be computed using SageMath \cite{sage}. For signed simple
graphs one can also compute the minimum distance using
Proposition~\ref{mar10-19} and the algorithms of 
\cite{rth-footprint,hilbert-min-dis}. For methods to calculate higher
weight enumerators of linear codes see \cite{Britz-Britz} and the references therein.  

\section{Matroids and linear codes}\label{matroids-codes}
A \textit{matroid} is a pair $M=(E,\rho)$ where $E$ is a finite set, called the
\textit{ground set} of $M$, and $\rho \colon 2^E\rightarrow
\mathbb{N}_0:=\{0,1,\ldots\}$ is a function, called the \textit{rank
function} of $M$, satisfying:
 \begin{itemize}
 \item[($\mathrm{R}_0$)] $\rho(\emptyset)=0$;
 \item[($\mathrm{R}_1$)] If $X\subseteq E$ and $e\in E$, then
 $\rho(X)\leq\rho(X\cup\{e\})\leq\rho(X)+1$;
 \item[($\mathrm{R}_2$)] If $X\subseteq E$ and $Y\subseteq E$, then 
$\rho(X \cup  Y)+\rho(X \cap Y) \leq \rho(X)+\rho(Y)$.
 \end{itemize}
\quad An \textit{independent set} of a matroid $M$ is subset $X\subseteq E$ such that
$\rho(X)=|X|$. In particular the empty set is an independent set. 
A \textit{base} is a maximal
independent set. A subset of the ground set which is not independent
is called 
\textit{dependent} and a \textit{circuit} of $M$ is a minimal dependent set.
We denote by $\mathcal{C}_M$ the family of all circuits
of $M$. The \textit{rank} of the matroid $M$, denoted $\rho(M)$,  
is $\rho(E)$. The \textit{nullity} of
$X\subseteq E$, denoted $\eta(X)$, is defined by
\[\eta(X):=|X|-\rho(X),\]
and the \textit{nullity} of $M$, denoted $\eta(M)$, is $\eta(E)$.  
 Let $M=(E,\rho)$ be a matroid. Its \textit{dual} is the matroid $M^*=(E,\rho^*)$
with the same ground set $E$ and rank function given by
\[\rho^*(X):=|X|-\rho(E)+\rho(E\backslash X),\quad X\subseteq E,\]
see \cite[p.~72]{oxley}. The nullity function of $M^*$ is denoted by
$\eta^*$. One can verify that $(M^*)^*=M$.

A family of circuits $\{C_i\}_{i=1}^r$ of a matroid $M$ is called
\textit{non-redundant} if $\bigcup_{i\neq j}C_i\subsetneq\bigcup_{i=1}^rC_i$
for $j=1,\ldots,r$ \cite{Britz-bond-cycle,Brylawski}. 
Let $X$ be a subset of the ground set $E$. The
\textit{degree} or \textit{non-redundancy} of $X$ is the maximum number of non-redundant
circuits contained in $X$, and it is denoted by $\deg(X)$. 

\begin{lemma}{\cite[p.~306, Table~A.1(6)]{Brylawski}}\label{JV1}
Let $M=(E,\rho)$ be a matroid,  let $X$ be a subset of $E$, and
let $\eta$ be the nullity function of $M$. Then  $\deg(X)=\eta(X)$.
\end{lemma}

\begin{theorem}{\cite[Theorem~2]{Wei}}\label{Wei-theorem}
Let $C$ be an $[m,k]$-linear code and let $M^*$ 
be the dual of the vector matroid of $C$. 
Then, the $r$-th generalized Hamming weight of $C$ is given by 
$$\delta_r(C)=\min\{|X|\colon
X\subseteq E\mbox{ and }\eta^*(X)\geq r\} \mbox{ for }1\leq
r\leq k.
$$
\end{theorem}

By Lemma~\ref{JV1}, we can replace the inequality $\eta^*(X)\geq r$ 
by $\eta^*(X)=r$. This result suggests how to 
define the generalized Hamming weights of
any matroid $M$.  

\begin{definition}{\cite[p.~332]{BJMS}} Let $M=(E,\rho)$ be a
matroid with nullity function $\eta$. The \textit{generalized Hamming
weights} of $M$ are defined as
\[d_r(M):=\min\{|X|\colon X \subseteq E\mbox{ and }\eta(X)=r\}\ \mbox{ for
}\ 1\leq r\leq \eta(E).\]
\end{definition}

\begin{lemma}\label{d-delta-relations} 
Let $C$ be a linear code of length $m$ and dimension $k$ and let $M$
be its vector matroid. Then $\delta_r(C)=d_r(M^*)$ for $1\leq
r\leq k$ and $\delta_r(C^\perp)=d_r(M)$ 
for $1\leq r\leq m-k$. 
\end{lemma}
\begin{proof} By Lemma~\ref{JV1} and Theorem~\ref{Wei-theorem}, 
we obtain $\delta_r(C)=d_r(M^*)$ for $1\leq r\leq k$. 
The matroid associated to $C^\perp$ is $M^*$. Hence
$\delta_r(C^\perp)=d_r(M)$ 
for $1\leq r\leq m-k$. 
\end{proof}

\begin{theorem}{\rm(\cite[Corollary~1.3]{tohaneanu-matroid},
\cite[Proposition~6]{Johnsen})}\label{formulas-for-GHWs-matroid} 
Let $M=(E,\rho)$ be a matroid and let $\eta$ be its nullity
function. The following hold.
\begin{align*}
&d_r(M^*)=\min\{|X|\colon
X\subseteq E\mbox{ and }\rho(E\setminus X)=\rho(E)-r\} \mbox{ for }1\leq
r\leq \rho(E).\\
&d_r(M)=\textstyle\min\left\{ 
\left|\bigcup_{i=1}^r C_i\right|\colon 
\{C_i\}_{i=1}^r\mbox{ are non-redundant circuits of } M\right\}\mbox{
for }1\leq r\leq \eta(E). 
\end{align*}
\end{theorem}
\begin{proof} 
According to 
\cite[Theorem~2, p.~35]{welsh}, one has $\rho(E\setminus
X)=\rho(E)-\eta^*(X)$. Therefore the first equality follows from
\begin{equation*}
d_r(M^*)=\min\{|X|\colon
X\subseteq E\mbox{ and }\eta^*(X)=r\} \mbox{ for }1\leq
r\leq \rho(E).
\end{equation*}
\quad On the other hand, recall that by definition of $d_r(M)$,
one has 
$$
d_r(M)=\min\{|X|\colon
X\subseteq E\mbox{ and }\eta(X)=r\} \mbox{ for }1\leq
r\leq\eta(E).
$$
\quad Therefore, applying Lemma~\ref{JV1}, the second equality
follows.  
\end{proof}

\begin{corollary}\label{formulas-for-GHWs} 
Let $C$ be an $[m,k]$-linear code and let
$M=(E,\rho)$ be the vector matroid of $C$. Then the following
equalities hold:
\begin{align*}
&\delta_r(C)=\min\{|X|\colon
X\subseteq E\mbox{ and }\rho(E\setminus X)=\rho(E)-r\} \mbox{ for }1\leq
r\leq k.\\
&\delta_r(C^\perp)=\textstyle\min\left\{ 
\left|\bigcup_{i=1}^r C_i\right|\colon 
\{C_i\}_{i=1}^r\mbox{ are non-redundant circuits of } M\right\}\mbox{
for }1\leq r\leq m-k. 
\end{align*}
\end{corollary}
\begin{proof} 
By Lemma~\ref{d-delta-relations}, 
we obtain $\delta_r(C)=d_r(M^*)$ for $1\leq r\leq k$ and 
$\delta_r(C^\perp)=d_r(M)$ 
for $1\leq r\leq m-k$. Thus the result follows from
Theorem~\ref{formulas-for-GHWs-matroid}.
\end{proof}

The number in the right hand side of the second equality of
Corollary~\ref{formulas-for-GHWs} is the
$r$-th \textit{circuit number} of $M$ and is denoted $\gamma_r(M)$
\cite{Britz-bond-cycle}. The 
$r$-th \textit{cocircuit number} $\beta_r(M)$ is defined similarly.

If $C$ is an $[m,k]$-linear code, then
$\delta_1(C)<\cdots<\delta_k(C)$ \cite{Johnsen-Roksvold-Verdure,Wei}. The following
\textit{duality} theorem of Wei is a classical result in
this area.

\begin{theorem}{\rm(Wei's duality
\cite[Theorem~3]{Wei})}\label{wei-duality} Let $C$ be
an $[m,k]$-linear code. Then 
$$\{\delta_r(C)\,\vert\,r=1,\ldots,k\}=\{1,\ldots,m\}\setminus\{m+1-
\delta_r(C^\perp)\,\vert\, r=1,\ldots,m-k\}.
$$
\end{theorem}

This result was generalized by  Britz, Johnsen, Mayhew and 
Shiromoto \cite[Theorem~5]{BJMS} from linear codes
to arbitrary matroids.

\section{Generalized Hamming weights over signed graphs}\label{GHW-section}

In this section we present our main results on linear codes. 
To avoid repetitions, we continue to employ
the notations and 
definitions used in Sections~\ref{intro-section} and \ref{matroids-codes}.

A multigraph $G$ consists of a finite set of vertices, $V(G)$, and a
finite multiset of edges, $E(G)$. Edges of $G$ are of two types. A
\textit{link} $e=\{v,w\}$, with two distinct endpoints, $v,w$ in $V(G)$ and a 
\textit{loop}, $e=\{v,v\}$, with two coincident endpoints. 
As $E(G)$ is a multiset, multiple edges are allowed. 
The number of edges of $G$ counted with multiplicity is denoted by
$m=|E(G)|$. A multigraph with no loops or multiple edges is called a
\textit{simple graph} or a \textit{graph}.
Let $G$ be a multigraph. A \textit{cycle} of $G$ is a simple closed
path in $G$. A loop is a cycle of length $1$, a pair of parallel
links is a cycle of length $2$,  
a triangle is a cycle of length 3, and so on. 
A maximal connected subgraph of a graph is called a 
\textit{connected component} of the graph. 

\begin{theorem}
{\rm(\cite[Theorem~5.1]{Zaslavsky-signed-graphs}, 
\cite[Theorem~2.1]{Zaslavsky-biased-graphs-II})}
\label{signed-graphic-matroid}
Let $G_\sigma$ be a signed graph. Then there exists a 
matroid $M(G_\sigma)$ on $E(G)$ whose circuits are the balanced cycles
and the bowties of $G_\sigma$.
\end{theorem}

The matroid $M(G_\sigma)$ is called the \textit{signed-graphic
matroid} of $G_\sigma$. The circuits of $M(G_\sigma)$ are called the 
\textit{circuits} of $G_\sigma$. If $G_\sigma$ is balanced, then
$M(G_\sigma)$ is the graphic matroid $M(G)$ of $G$.

\begin{definition}
Let $G_\sigma$ be an unbalanced (resp. balanced) signed graph. 
A \textit{cutset} of $G_\sigma$ is a set of edges
whose removal from $G_\sigma$ increases the number of balanced
connected components (resp. connected components) of $G_\sigma$. 
A \textit{cocircuit} of $G_\sigma$ is a
minimal cutset of $G_\sigma$. If $G_\sigma$ is balanced, a cocircuit
of $G_\sigma$ is called a \textit{cocycle} or \textit{bond} of $G$.
\end{definition}

\begin{lemma}\cite[Proposition~2.3.1]{oxley}\label{cocycles-graphic-matroid}
If $G_\sigma$ is a balanced signed graph, then the cocircuits of
$G_\sigma$ are the cocircuits of the graphic matroid $M(G)$ of $G$,
that is, the circuits of $M(G)^*$.
\end{lemma}

\begin{lemma}\cite[Theorem~5.1(i)]{Zaslavsky-signed-graphs}\label{jan15-19}
If $G_\sigma$ is a connected unbalanced signed graph and $M(G_\sigma)$
is its signed-graphic matroid, then the cocircuits of $G_\sigma$ are
the cocircuits of $M(G_\sigma)$, that is, the circuits of
the dual matroid $M(G_\sigma)^*$.
\end{lemma}

\begin{proof} Let $\rho$ be the rank function of 
$M(G_\sigma)$. We set $V=V(G)$ and $E=E(G)$. Take a cocircuit
$X\subseteq E$ of $G_\sigma$. As $G_\sigma$ is 
connected, one has $c_0(E)=0$ and $c_0(E\setminus
X)=1$ if $G_\sigma$ is unbalanced, and $c_0(E)=1$ and 
$c_0(E\setminus X)=2$ if $G_\sigma$ is balanced. Then 
$c_0(E\setminus X)=c_0(E)+1$. According to 
\cite[Theorem~5.1(j)]{Zaslavsky-signed-graphs}, one has
\begin{equation}\label{mar15-19} 
\rho(E)=|V|-c_0(E)\ \mbox{ and }\ \rho(E\setminus
X)=|V|-c_0(E\setminus X).
\end{equation}
\quad Therefore $\rho(E\setminus X)=\rho(E)-1$. Since $X$ is a
minimal cutset, it follows that $H:=E\setminus X$ is closed, that is,
$\rho(H\cup\{e\})=\rho(H)+1=\rho(E)$ for each $e\notin H$. Indeed, from the
equality $H\cup\{e\}=E\setminus(X\setminus\{e\})$, and the minimality
of $X$, we get $c_0(H\cup\{e\})=c_0(E)$. Hence, using
Eq.~(\ref{mar15-19}), we obtain $\rho(H\cup\{e\})=\rho(E)=\rho(H)+1$
for each $e\notin H$. As a consequence, $H$ is a maximal set of rank
$\rho(E)-1$. Thus, by \cite[Lemma~1, p.~38]{welsh}, $H$ is a
hyperplane of $M(G_\sigma)$ in the sense of \cite{welsh}, and by \cite[Theorem~2,
p.~39]{welsh}, $X$ is a cocircuit of $M(G_\sigma)$. Similarly, if $X$
is a cocircuit of $M(G_\sigma)$, it is seen that $X$ is a cocircuit 
of $G_\sigma$.
\end{proof}
 
\begin{lemma}\label{GHW-signed-graphic-matroid} 
Let $G_\sigma$ be a connected signed graph, let $\rho$ and
$\eta$ be the rank and nullity functions of the signed-graphic 
matroid $M=M(G_\sigma)$ of $G_\sigma$. 
The following hold.
\begin{itemize}
\item[(i)] If $1\leq r\leq \eta(M)$, then $d_r(M)$ is equal to the
minimum number of edges of $G$ forming a union of $r$
non-redundant balanced cycles and bowties of $G_\sigma$. 
\item[(ii)] If $1\leq r\leq\rho(M)$ and $G_\sigma$ is unbalanced
$($resp. balanced$)$, then
$d_r(M^*)$ is the $r$-th cogirth $\upsilon_r(G_\sigma)$ $($resp. 
$r$-th edge connectivity $\lambda_r(G_\sigma)$$)$ of $G_\sigma$.  
\item[(iii)] If $1\leq r\leq\rho(M)$, then $d_r(M^*)$ is equal to the
minimum number of edges of $G_\sigma$ forming a union of $r$
non-redundant cocircuits of $G_\sigma$. 
\end{itemize}
\end{lemma}

\begin{proof} (i): By Theorem~\ref{signed-graphic-matroid}, the
circuits of $M$ are the balanced cycles and the bowties of $G_\sigma$. Hence, it suffices
to recall the following formula of 
Theorem~\ref{formulas-for-GHWs-matroid}:
$$
d_r(M)=\textstyle\min\left\{ 
\left|\bigcup_{i=1}^r C_i\right|\colon 
\{C_i\}_{i=1}^r\mbox{ are non-redundant circuits of } M\right\}\mbox{
for }1\leq r\leq \eta(M). 
$$
\quad (ii): Let $E$ be the edge set of $G_\sigma$, which is the ground set of $M$, and let
$V$ be the vertex set of $G_\sigma$. According to
\cite[Theorem~5.1(j)]{Zaslavsky-signed-graphs} the rank function of
$M(G_\sigma)$ satisfies
\begin{equation}\label{mar9-19} 
\rho(E\setminus X)=|V|-c_0(E\setminus X)\ \mbox{ for }\ X\subseteq E, 
\end{equation}
where
$c_0(E\setminus X)$ is the number of balanced connected components of
the signed subgraph $G_\sigma\setminus X$ with edge set $E\setminus
X$ 
and vertex set $V$. Therefore, by
Theorem~\ref{formulas-for-GHWs-matroid}, 
we obtain
\begin{eqnarray*}
d_r(M^*)&=&\min\left\{ 
\left|X\right|\colon 
X\subseteq E\mbox{ and }\rho(E\setminus X)=\rho(E)-r\right\}\\
&=&\min\left\{ 
\left|X\right|\colon 
X\subseteq E\mbox{ and }|V|-c_0(E\setminus X)=\rho(E)-r\right\}
\end{eqnarray*}
for $1\leq r\leq\rho(E)$. If $G_\sigma$ is unbalanced (resp.
balanced), then by making
$X=\emptyset$ in Eq.~(\ref{mar9-19}) we get $\rho(E)=|V|$
(resp. $\rho(E)=|V|-1$). Therefore
$$
d_r(M^*)=\begin{cases}
\min\left\{ 
\left|X\right|\colon 
X\subseteq E\mbox{ and }c_0(E\setminus X)=r\right\}=\upsilon_r(G_\sigma)&\mbox{if
}G_\sigma\mbox{ is unbalanced},\\
\min\left\{ 
\left|X\right|\colon 
X\subseteq E\mbox{ and }c(E\setminus X)=r+1\right\}=\lambda_r(G_\sigma)
& \mbox{if
}G_\sigma\mbox{ is balanced},
\end{cases}
$$
where $c(E\setminus X)$ is the number of connected components of the
signed subgraph $G_\sigma\setminus X$.
 
(iii): By Lemmas~\ref{cocycles-graphic-matroid} and
\ref{jan15-19}, 
the circuits of the dual matroid
$M^*$ of $M$ are the
cocircuits of the signed graph $G_\sigma$, and by 
Theorem~\ref{formulas-for-GHWs-matroid} we get
$$
d_r(M^*)=\textstyle\min\left\{ 
\left|\bigcup_{i=1}^r C_i^*\right|\colon 
\{C_i^*\}_{i=1}^r\mbox{ are non-redundant circuits of } M^*\right\}\mbox{
for }1\leq r\leq \eta^*(E).
$$
\quad Hence, the required equality follows noticing that $\eta^*(E)=\rho(E)$. 
\end{proof}

\begin{definition}\label{incidence-matrix} 
Let $G_\sigma$ be a signed graph with $s$ vertices $t_1,\ldots,t_s$
and $m$
edges, let
$K$ be a field, and let $\mathbf{e}_i$ be the $i$-th unit vector in
$K^s$. The \textit{incidence matrix} of $G_\sigma$ over the field $K$
is the $s\times m$ matrix $A$ whose column vectors are given by:
\begin{itemize}
\item[(i)] $\mathbf{e}_i-\mathbf{e}_j$ (resp. $\mathbf{e}_i+\mathbf{e}_j$) if $e=\{t_i,t_j\}$ is a
link with $\sigma(e)=+$ (resp. $\sigma(e)=-$);
\item[(ii)] $\mathbf{0}$ (resp. $2\mathbf{e}_i$) if $e=\{t_i,t_i\}$ is a loop with $\sigma(e)=+$
(resp. $\sigma(e)=-$).
\end{itemize}
\end{definition}

Note that the columns of $A$ are defined up to sign, so one can 
pick $\mathbf{e}_i-\mathbf{e}_j$ or $\mathbf{e}_j-\mathbf{e}_i$ if $e=\{t_i,t_j\}$ is a
link with $\sigma(e)=+$. To avoid ambiguity we could normalize and pick 
$e_i-e_j$ if $i>j$. The order of the columns of $A$ and the choice of
sign have no significance for the invariants of linear codes, signed
graphs, and Stanley--Reisner ideals that we want to
study.

If $G$ is a multigraph with vertices $t_1,\ldots,t_s$, the
\textit{incidence matrix} 
of $G$ over a field $K$ is the incidence matrix of the negative signed
graph $G_-$, that is, the matrix whose
columns are all vectors $\mathbf{e}_i+\mathbf{e}_j$ such that
$\{t_i,t_j\}$ is an edge of $G$. 
A \textit{digraph} $\mathcal{D}$ consists of a multigraph
$G$ with vertices $t_1,\ldots,t_s$ where all edges of $G$ are directed from
one vertex to another. The \textit{edges} or \textit{arrows} of $\mathcal{D}$ 
are ordered pairs of vertices $(t_i,t_j)$ with $e=\{t_i,t_j\}$ an edge
of $G$, where $(t_i, t_j)$ represents the edge $e$ directed
from $t_i$ to $t_j$. The \textit{incidence matrix} of $\mathcal{D}$ over a field
$K$ is the incidence matrix of the positive signed graph $G_+$, that
is, the matrix whose
columns are all vectors $\mathbf{e}_i-\mathbf{e}_j$ such that
$(t_i,t_j)$ is an edge of $\mathcal{D}$. 

\begin{theorem}\cite[Theorems~8B.1, 8B.2]{Zaslavsky-signed-graphs}
\label{signed-graphic-rep}
Let $G_\sigma$ be a signed graph and let $A$ be
its incidence matrix over a field of characteristic $p$. The following
hold.
\begin{itemize}
\item[(a)] If $p\neq 2$, then the vector matroid $M[A]$ of $A$ is the
signed-graphic matroid $M(G_\sigma)$. 
\item[(b)] If $p=2$, then 
$M[A]$ is the graphic matroid $M(G)$ of $G$. 
\end{itemize}
\end{theorem}

\begin{proposition}
\label{rank-incidence-signed}
Let $G_\sigma$ be a signed graph with $s$ vertices, $c$ connected components, 
$c_0$ balanced connected components, and
let $A$ be its incidence matrix over a field $K$. Then 
$$
{\rm rank}(A)=\begin{cases} s-c_0&\text{if }{\rm char}(K)\neq 2
,\\
s-c&\text{if }{\rm char}(K)=2 \text{ or }G_\sigma\text{ is balanced}.
\end{cases}
$$
\end{proposition}

\begin{proof} Assume ${\rm char}(K)\neq 2$. By
Theorem~\ref{signed-graphic-rep}(a), the signed graphic matroid
$M(G_\sigma)$ is the vector matroid $M[A]$.  According to 
\cite[Theorem~5.1(j)]{Zaslavsky-signed-graphs}, the rank of
$M(G_\sigma)$ is $s-c_0$. Thus in this case ${\rm rank}(A)=s-c_0$.
Assume ${\rm char}(K)=2$. By
Theorem~\ref{signed-graphic-rep}(b), the graphic matroid $M(G)$ is the
vector matroid $M[A]$. If $G$ is connected,
then the bases of the matroid $M(G)$ are the spanning trees of $G$
\cite[p.~28]{welsh}, and ${\rm
rank}(A)=s-1$. As a consequence, if $G$ has $c$ components, one has 
${\rm rank}(A)=s-c$. If $G_\sigma$ is balanced, then $c=c_0$, and by
the previous two cases ${\rm rank}(A)=s-c$, regardless of the
characteristic of the field $K$.
\end{proof}

\begin{corollary}\label{rank-incidence-oriented} 
Let $\mathcal{D}$ be a digraph with $s$ vertices and $c$ 
connected components, and
let $A$ be its incidence matrix over a field $K$. Then, ${\rm
rank}(A)=s-c$. 
\end{corollary}

\begin{proof} Let $G$ be the underlying unoriented simple graph of
$\mathcal{D}$. Consider the positive signed graph $G_+$. Note that $G_+$ is
balanced. Since $\mathcal{D}$ and $G_+$ have the same incidence matrix,
the result follows
from Proposition~\ref{rank-incidence-signed}.
\end{proof}

\begin{definition}
The \textit{incidence matrix code} of a signed graph $G_\sigma$ (resp.
multigraph $G$, digraph $\mathcal{D}$), over a finite field
$\mathbb{F}_q$, is the
linear code $C$ generated by the rows of the incidence matrix of
the signed graph $G_\sigma$ (resp. multigraph $G$, digraph $\mathcal{D}$). 
\end{definition}

\begin{corollary}\label{dim-cp} Let $G_\sigma$ be a connected signed graph with $s$ vertices and 
$m$ edges, and let $C$ be the incidence matrix code of 
$G_\sigma$ over a finite field of characteristic $p$. Then
\begin{itemize}
\item[(a)] $C$ $($resp. $C^\perp$$)$ is an $[m,s]$ $($resp.
$[m,m-s]$$)$ linear code if $p\neq 2$
and $G_\sigma$ is unbalanced. 
\item[(b)] $C$ $($resp. $C^\perp$$)$ is an $[m,s-1]$ $($resp.
$[m,m-s+1]$$)$ linear code if $p=2$ or $G_\sigma$ is balanced. 
\end{itemize}
\end{corollary}

\begin{proof} This follows from
Proposition~\ref{rank-incidence-signed} noticing that $\dim(C)+\dim(C^\perp)=m$.
\end{proof}

\begin{definition}\label{bowtie} Let $G$ be a multigraph. A \textit{bowtie} 
of $G$ is the union of two odd cycles which meet at a single
vertex or the union of two vertex-disjoint odd cycles 
and a simple path which meets one cycle at each end and is otherwise
disjoint from them.
\end{definition}

\begin{corollary}\label{circuits-signed-graphic} Let $G$ be a
multigraph and let $G_+$ and $G_-$ be the positive and negative signed
graphs, respectively. The following hold.
\begin{itemize}
\item[(a)] The circuits of the signed-graphic matroid $M(G_+)$ are the cycles of $G$, that is,
$M(G_+)$ is the graphic matroid $M(G)$ of $G$.
\item[(b)] The signed-graphic matroid $M(G_+)$ is the vector matroid,
over any field $K$, of the incidence matrix of $G_+$ whose columns are of the form
$\mathbf{e}_i-\mathbf{e}_j$.
\item[(c)] The balanced $($resp. unbalanced$)$ cycles of $G_-$ are the even
$($resp. odd$)$ cycles of $G$. 
A circuit of $M(G_-)$ is 
either an even cycle or a bowtie of $G$.
\item[(d)] If $G_\sigma$ is a balanced signed graph, then
$M(G_\sigma)$ is the graphic matroid $M(G)$ of $G$. 
\end{itemize}
\end{corollary}

\begin{proof} (a): There are no unbalanced cycles of $G_+$. Hence, by
Theorem~\ref{signed-graphic-matroid}, the circuits of $M(G_+)$ are the
cycles of $G$.

(b): Let $p$ be the characteristic of the field $K$. If $p\neq 2$, by 
Theorem~\ref{signed-graphic-rep}, $M(G_+)$ is the
vector matroid of the incidence matrix $A$ of $G_+$ and the columns of
this matrix have the required form. If $p=2$, the graphic matroid $M(G)$ of $G$ is
the vector matroid $M[A]$ of the incidence matrix $A$ of $G$ 
\cite[Theorem~3, p.~149]{welsh}. By part (a), $M(G_+)$ is the graphic
matroid of $G$. Thus $M(G_+)$ is the vector matroid of $A$. The columns
of $A$ have the required form because in this case $1=-1$.

(c), (d): These follow readily from Theorem~\ref{signed-graphic-matroid}.
\end{proof}

\begin{remark} Let $G$ be a multigraph and let $A$ be the incidence
matrix of $G_+$ over the field $K=\mathbb{Q}$ of rational numbers. 
Since $M[A]$ is the graphic matroid of $G$, 
to compute all cycles of $G$ one can use \textit{Macaulay}$2$
\cite{mac2} and the package \textit{Matroids} \cite{matroids}.  
\end{remark}

\begin{corollary}{\cite{Simoes-Pereira,Vi3,Zaslavsky-signed-graphs}}
\label{VectorCircuit}
If $A$ is the incidence matrix
of a multigraph $G$ over a field of ${\rm char}(K)\neq 2$, 
then the circuits of the vector matroid $M[A]$ are the even cycles 
and bowties of $G$. 
\end{corollary}
\begin{proof} It follows from
Theorem~\ref{signed-graphic-rep}(a) and
Corollary~\ref{circuits-signed-graphic}(c) 
by considering $G_-$.
\end{proof}
Our main results on linear codes are the following. First, we give
graph theoretical formulas for the GHWs for the incidence matrix code
of a signed graph.

\begin{theorem}\label{pepe-vila-2019} Let $C$ be the
incidence matrix code of a connected signed graph
$G_\sigma$ with $s$ vertices, $r$-th cogirth $\upsilon_r(G_\sigma)$, 
$r$-th edge connectivity  
$\lambda_r(G)$, over a finite field $K$ of ${\rm
char}(K)=p$. Then, the $r$-th generalized Hamming weight of $C$ is given by 
$$ 
\delta_r(C)=\begin{cases}
\upsilon_r(G_\sigma)&\text{if } p\neq 2,\, G_\sigma\textit{ is
unbalanced and }\, 1\leq r\leq s,\\
\lambda_r(G)&\text{if }p=2\mbox{ and }1\leq r\leq s-1,\\
\lambda_r(G)&\text{if }G_\sigma\textit{ is balanced and }\, 1\leq r\leq s-1.
\end{cases}
$$
\end{theorem}
\begin{proof} Let $A$ be the incidence matrix of $G_\sigma$ and let
$\rho$ be the rank function of the vector matroid $M=M[A]$. According 
to Proposition~\ref{rank-incidence-signed}, $\rho(M)=s$ if
$p\neq 2$ and $G_\sigma$ is unbalanced, and $\rho(M)=s-1$ if $p=2$
or $G_\sigma$ is balanced.

Assume that $p\neq 2$. By Theorem~\ref{signed-graphic-rep}(a), the
signed-graphic matroid $M(G_\sigma)$ is the vector matroid $M=M[A]$.
Hence, using Lemmas~\ref{d-delta-relations} and \ref{GHW-signed-graphic-matroid}(ii),  
one has
$$ 
\delta_r(C)=d_r(M^*)=\begin{cases}
\upsilon_r(G_\sigma)&\text{if } G_\sigma\textit{ is unbalanced and
}\, 1\leq
r\leq s,\\
\lambda_r(G)&\text{if }G_\sigma\textit{ is balanced and }\, 1\leq
r\leq s-1.
\end{cases}
$$
\quad Assume that $p=2$. By Theorem~\ref{signed-graphic-rep}(b), $M=M[A]$ is the graphic matroid
$M(G)$ and, by Corollary~\ref{circuits-signed-graphic}(a), $M(G_+)$ is 
also the graphic matroid $M(G)$. As $M(G_+)$ is balanced, by 
Lemmas~\ref{d-delta-relations} and \ref{GHW-signed-graphic-matroid}(ii), we get
$\delta_r(C)=d_r(M^*)=\lambda_r(G_+)=\lambda_r(G)$.
\end{proof}

Let $G$ be a multigraph. The $r$-th \textrm{cogirth} 
$\upsilon_r(G_-)$ of $G_-$ is the minimum number of edges whose removal
results in a multigraph with $r$ bipartite connected components. If $r=1$,
$\upsilon_1(G_-)$ is denoted $\upsilon(G_-)$. For simple graphs, the
following combinatorial formulas for the  
generalized Hamming weights were shown in \cite{linear-codes}. 

\begin{corollary}\label{pepe-vila-2018-coro1} Let $C$ be the
incidence matrix code of a connected multigraph
$G$ with $s$ vertices
over a finite field $K$ of ${\rm char}(K)=p$. Then  
$$ 
\delta_r(C)=\begin{cases}
\upsilon_r(G_-)&\text{if } p\neq 2,\, G\textit{ is non-bipartite and
}\, 1\leq r\leq s,\\
\lambda_r(G)&\text{if }p=2\mbox{ and }\, 1\leq r\leq s-1,\\
\lambda_r(G)&\text{if }G\textit{ is bipartite and }\, 1\leq r\leq s-1.
\end{cases}
$$
\end{corollary}

\begin{proof} It follows from Theorem~\ref{pepe-vila-2019} by
considering the negative signed graph $G_-$ and noticing that $G_-$ is
balanced if and only if $G$ is bipartite. 
\end{proof}

The next result shows combinatorial formulas for the minimum
distance of the incidence matrix code of a signed graph 
\cite[Proposition~9.2.4]{oxley}.

\begin{corollary}\label{Dankelmann-etal} 
Let $C$ be the incidence matrix code of a connected signed graph
$G_\sigma$ with $s$ vertices, cogirth $\upsilon(G_\sigma)$, edge connectivity  
$\lambda(G_\sigma)$, over a finite field $K$ of ${\rm
char}(K)=p$. Then, the minimum distance $\delta(C)$ of $C$ is given by 
$$ 
\delta(C)=\begin{cases}
\upsilon(G_\sigma)&\text{if } p\neq 2,\, G_\sigma\textit{ is
unbalanced and }\, 1\leq r\leq s,\\
\lambda(G_\sigma)&\text{if }p=2\mbox{ and }\, 1\leq r\leq s-1,\\
\lambda(G_\sigma)&\text{if }G_\sigma\textit{ is balanced and }\, 1\leq r\leq s-1.
\end{cases}
$$
\end{corollary}

\begin{proof} It follows by making $r=1$ in Theorem~\ref{pepe-vila-2019}.
\end{proof}

Our next result gives graph theoretical formulas for the
generalized Hamming weights of the
dual code of the incidence matrix code of a signed graph. 

\begin{theorem}\label{pepe-vila-2019-dual} Let $G_\sigma$ be a
connected signed graph with $s$ vertices and
$m$ edges, and let $C$ be the incidence matrix code of $G_\sigma$ over a finite field
$K$ of characteristic $p$. The following hold.
\begin{itemize}
\item[(a)] If $p=2$ or $G_\sigma$ is balanced, 
and $1\leq r\leq m-s+1$ $($resp. $1\leq r\leq s-1$$)$, then
$\delta_r(C^\perp)$ $($resp. $\delta_r(C)$$)$ is the minimum number
of edges of $G$ forming a union of $r$
non-redundant cycles $($resp. cocycles$)$ of $G$. 
\item[(b)] If $p\neq 2$ and $1\leq r\leq m-s$ $($resp. $1\leq r\leq s$$)$, then
$\delta_r(C^\perp)$ $($resp. $\delta_r(C)$$)$ is the 
minimum number of edges of $G$ forming a union of $r$
non-redundant balanced cycles and bowties $($resp. cocircuits$)$ of $G_\sigma$. 
\end{itemize}
\end{theorem}

\begin{proof} (a): Assume $p=2$ and $1\leq r\leq m-s+1$. Let $A$ be 
the incidence matrix of $G_\sigma$. By 
Theorem~\ref{signed-graphic-rep}(b), the vector 
matroid $M=M[A]$ is the graphic matroid $M(G)$. Thus the circuits of
$M[A]$ are the cycles of $G$. Therefore, by
Corollary~\ref{formulas-for-GHWs}, 
we get 
$$
\delta_r(C^\perp)=\textstyle\min\left\{ 
\left|\bigcup_{i=1}^r C_i\right|\colon 
\{C_i\}_{i=1}^r\mbox{ are non-redundant cycles of } G\right\}. 
$$
\quad Assume $p=2$ and $1\leq r\leq s-1$. By the previous part,
$M=M[A]$ is the graphic matroid $M(G)$. The circuits of $M^*=M[A]^*$ are the cocycles of $G$
\cite[p.~41]{welsh}, that is, these are edge sets $X$ whose removal
from $G$ increases the number of connected components of $G$ and are
minimal with respect to this property. The 
dual matroid $M^*$ of $M$ is the vector matroid of $C^\perp$
\cite[p.~141]{welsh}. Therefore, by
Corollary~\ref{formulas-for-GHWs} and the equality $C=(C^\perp)^\perp$, 
we get
\begin{eqnarray*}
\delta_r(C)&=&\textstyle\min\left\{ 
\left|\bigcup_{i=1}^r C_i^*\right|\colon 
\{C_i^*\}_{i=1}^r\mbox{ are non-redundant circuits of } M^*\right\}\\
&=&\textstyle\min\left\{ 
\left|\bigcup_{i=1}^r C_i^*\right|\colon 
\{C_i^*\}_{i=1}^r\mbox{ are non-redundant cocycles of } G\right\}. 
\end{eqnarray*}
\quad Assume that $G_\sigma$ is balanced. If $p=2$, the formulas for
$\delta_r(C^\perp)$ and $\delta_r(C)$ follow from the two previous cases.
Assume $p\neq 2$. By Theorem~\ref{signed-graphic-rep}(a), the vector
matroid $M[A]$ is the signed-graphic matroid $M(G_\sigma)$ and, by 
Corollary~\ref{circuits-signed-graphic}(d), $M(G_\sigma)$ is the
graphic matroid $M(G)$. Hence, we can proceed as in the previous cases.  

\quad (b): Assume $1\leq r\leq m-s$. Let $A$ be 
the incidence matrix of $G_\sigma$. By
Theorem~\ref{signed-graphic-matroid}, the circuits of $M(G_\sigma)$
are the balanced cycles and bowties of $G_\sigma$. As $p\neq 2$, 
by Theorem~\ref{signed-graphic-rep}(a), the vector
matroid $M=M[A]$ is $M(G_\sigma)$. Therefore the circuits of $M[A]$ are
the balanced cycles and bowties of $G_\sigma$, and by
Corollary~\ref{formulas-for-GHWs} one has  
$$
\delta_r(C^\perp)=\textstyle\min\left\{ 
\left|\bigcup_{i=1}^r C_i\right|\colon 
\{C_i\}_{i=1}^r\mbox{ are non-redundant circuits of } M\right\}. 
$$

Assume $1\leq r\leq s$. As $M=M[A]$ is the signed-graphic
matroid $M(G_\sigma)$, the circuits of $M^*=M[A]^*$ are the cocircuits
of $G_\sigma$ by 
Lemmas~\ref{cocycles-graphic-matroid} and \ref{jan15-19}. The 
dual matroid $M^*$ of $M$ is the vector matroid of $C^\perp$
\cite[p.~141]{welsh}. Hence, by
Corollary~\ref{formulas-for-GHWs} and noticing $C=(C^\perp)^\perp$, 
we get
\begin{eqnarray*}
\delta_r(C)&=&\textstyle\min\left\{ 
\left|\bigcup_{i=1}^r C_i^*\right|\colon 
\{C_i^*\}_{i=1}^r\mbox{ are non-redundant circuits of } M^*\right\}\\
&=&\textstyle\min\left\{ 
\left|\bigcup_{i=1}^r C_i^*\right|\colon 
\{C_i^*\}_{i=1}^r\mbox{ are non-redundant cocircuits of }
G_\sigma\right\}.
\end{eqnarray*}
\quad This completes the proof of part (b).
\end{proof}

\begin{corollary}\label{pepe-vila-2018-digraph} Let $C$ be the
incidence matrix code, over a finite field $K=\mathbb{F}_q$, 
of a connected digraph $\mathcal{D}$ with $s$ vertices and $m$ edges, 
and let $G$ be its underlying multigraph. 
Then  
\begin{itemize}
\item[(a)] $
\delta_r(C)=\lambda_r(G)\mbox{ for } 1\leq r\leq s-1.
$
\item[(b)] If $1\leq r\leq m-s+1$ $($resp. $1\leq r\leq s-1$$)$, then 
$\delta_r(C^\perp)$ $($resp. $\delta_r(C)$$)$ 
is the minimum number of edges of $G$ forming a union of $r$
non-redundant cycles $($resp. cocycles$)$ of $G$. 
\end{itemize}
\end{corollary}

\begin{proof} Parts (a) and (b) follow from Theorems~\ref{pepe-vila-2019} and
\ref{pepe-vila-2019-dual}, respectively,  by considering the positive signed graph
$G_+$ and noticing that this is a balanced signed graph. 
\end{proof}

\begin{corollary}\label{feb4-19} Let $G$ be a connected multigraph with $s$ vertices and
$m$ edges, and let $C$ be its incidence matrix code over a finite field
$K=\mathbb{F}_q$ of characteristic $p$. The following hold.
\begin{itemize}
\item[(a)] If $p=2$ or $G$ is bipartite, and $1\leq r\leq m-s+1$ 
$($resp. $1\leq r\leq s-1$$)$, then 
$\delta_r(C^\perp)$ $($resp. $\delta_r(C)$$)$ 
is the minimum number of edges of $G$ forming a union of $r$
non-redundant cycles $($resp. cocycles$)$ of $G$. 
\item[(b)] If $p\neq 2$  and $1\leq r\leq m-s$, then
$\delta_r(C^\perp)$ is the 
minimum number of edges of $G$ forming a union of $r$
non-redundant even cycles and bowties of $G$. 
\end{itemize}
\end{corollary}

\begin{proof} (a): This follows from
Theorem~\ref{pepe-vila-2019-dual}(a) noticing that, if $G$ is
bipartite, then the circuits of $M(G_+)$ and $M(G_-)$ are the cycles 
of $G$. 

(b): This part follows from Corollary~\ref{circuits-signed-graphic}(c) and
Theorem~\ref{pepe-vila-2019-dual}(b), by considering the negative signed graph $G_-$.
\end{proof}

\section{The regularity of the ideal of
circuits}\label{regularity-section}

Let $M$ be a matroid on $E=\{1,\ldots,m\}$, let $\Delta$ be its independence
complex, that is, the faces of $\Delta$ are the independent
sets of $M$, and let $R=\mathbb{Q}[x_1,\ldots,x_m]=\bigoplus_{d=0}^\infty
R_d$ be a polynomial
ring with the standard grading over 
the field of rational numbers. It is convenient also to think of $E$
as the set of variables $\{x_1,\ldots,x_m\}$. The Stanley--Reisner
ideal $I_\Delta$ of $\Delta$, in the sense of \cite{Sta2}, 
is the \textit{edge ideal} $I(\mathcal{C}_M)$ of the
clutter of circuits $\mathcal{C}_M$ of $M$, that is, $I_\Delta$ is
the \textit{ideal of circuits} of $M$ generated by all 
squarefree monomials $\prod_{i\in X}x_i$ such that $X$ is a circuit of $M$. 

The simplicial complex $\Delta$
is pure shellable, in particular the ideal $I_\Delta$ is
Cohen--Macaulay, and the graded Betti
numbers of the Stanley--Reisner ring $K[\Delta]=R/I_\Delta$ are the
same if we replace $\mathbb{Q}$ by any
other field (see \cite[Remark~1, p.~78]{JohVer} and the 
references therein). 

\begin{definition}\label{regularity-socle-degree}\rm Let $I\subset R$ be a graded ideal and let
${\mathbf F}$ be the minimal graded free resolution of $R/I$ as an
$R$-module:
\[
{\mathbf F}:\ \ \ \textstyle 0\rightarrow
\bigoplus_{j}R(-j)^{\beta_{g,j}}
\stackrel{}{\rightarrow} \cdots
\rightarrow\bigoplus_{j}
R(-j)^{\beta_{1,j}}\stackrel{}{\rightarrow} R
\rightarrow R/I \rightarrow 0.
\]
\quad The $(i,j)$-th \textit{graded Betti number} of
$R/I$, denoted $\beta_{i,j}(R/I)$, is $\beta_{i,j}$, the integer $j$ is a
\textit{shift} of the resolution, $g$ is the
\textit{projective dimension} of $R/I$, and the \textit{regularity} of $R/I$ is 
$${\rm reg}(R/I):=\max\{j-i \mid \beta_{i,j}\neq 0\}.$$
\quad If $R/I$ is Cohen-Macaulay (i.e. $g=\dim(R)-\dim(R/I)$) and
there is a unique $j$ such that $\beta_{g,j}\neq 0$, then the ring $R/I$
is called {\em level}. 
\end{definition}

An excellent reference for the regularity of graded ideals and Betti
numbers is the book of Eisenbud \cite{eisenbud-syzygies}. The shifts and the Betti numbers of
the Stanley--Reisner ring of the independence complex $\Delta$ of a
matroid $M$ were determined by Johnsen and 
Verdure \cite{JohVer}.  

The following result shows that one can read the 
generalized Hamming weights of a matroid $M$ from the minimal graded free resolution of the
ideal of circuits of $M$. 

\begin{theorem}{\cite[Theorem~2]{JohVer}}\label{JV3} Let $M$ be a
matroid, let $R/I_\Delta$ be the Stanley--Reisner
ring of the independence complex $\Delta$ of $M$, and let 
$\beta_{r,j}(M)$ denote the $(r,j)$-th 
graded Betti number of $R/I_\Delta$. Then 
the generalized Hamming weights of $M$ are given by 
\[
d_r(M)=\min\{j:\beta_{r,j}(M)\neq 0\}\ \mbox{ for }\ 1\leq r\leq
\eta(M).
\]
\end{theorem}

\begin{corollary}\label{dec30-18} Let $C$ be an $[m,k]$-linear code and let
$R/I_\Delta$ be the Stanley--Reisner ring of the independence complex
of the matroid of $C$. Then
\[
\delta_r(C^\perp)=\min\{j:\beta_{r,j}(R/I_\Delta)\neq 0\}\ \mbox{ for }\ 1\leq r\leq
m-k.
\]
\end{corollary}

\begin{proof} It follows from Lemma~\ref{d-delta-relations} and
Theorem~\ref{JV3}.
\end{proof}

The following notion of a non-degenerate code will play a role here.
\begin{definition}
If $C \subseteq K^m$ is a linear code and $\pi_i$ is the $i$-th
projection map 
$$
\pi_i: C \rightarrow K,\ \ \ (v_1,\ldots,v_m)\mapsto v_i,
$$
for $i=1,\ldots,m$, we say that $C$ is \textit{degenerate} if for some $i$
the image of $\pi_i$ is zero, otherwise we say that $C$ is \textit{non-degenerate}.   
\end{definition}

\begin{remark}\label{dec29-18}
If $C\subseteq K^m$ is a non-degenerate linear code, then $\delta_k(C)=m$, where $k$ is the
dimension of $C$. If all columns of a generator matrix of
$C$ are non-zero, then $\pi_i\neq 0$ for $i=1,\ldots,m$ and $C$ is non-degenerate.  
\end{remark}
\begin{lemma}\label{dec28-18}
Let $M$ be the matroid on $E$ of a linear code $C$ and let $\Delta$
$($resp. $\Delta^*$$)$ be the 
independence complex of $M$ $($resp. $M^*$$)$. The following hold.
\begin{itemize}
\item[(a)] If $r=\dim(C^\perp)$, then ${\rm reg}(R/I_{\Delta})=\delta_r(C^\perp)-\dim(C^\perp)$. 
\item[(b)] If $r=\dim(C)$, then ${\rm reg}(R/I_{\Delta^*})=\delta_r(C)-\dim(C)$.
\item[(c)] If $C^\perp$ is non-degenerate, 
then ${\rm reg}(R/I_{\Delta})=\dim(C)$. 
\item[(d)] If $C$ is non-degenerate, 
then ${\rm reg}(R/I_{\Delta^*})=\dim(C^\perp)$. 
\end{itemize}
\end{lemma}

\begin{proof} By  \cite[Corollary~6.3.5]{monalg-rev}, the 
Stanley--Reisner ring $K[\Delta]:=R/I_\Delta$ has Krull dimension
$\dim(\Delta)+1$. As $\dim(\Delta)$ is $\rho(M)-1$ and $\dim(R)=|E|$,  one has $|E|-{\rm
ht}(I_\Delta)=\rho(M)$, where ${\rm
ht}(I_\Delta)$ is the height of the ideal $I_\Delta$. Therefore
\begin{eqnarray*}
\eta(M)&=&|E|-\rho(M)={\rm ht}(I_\Delta)\\
&=&|E|-\dim(C)=\dim(C^\perp)=\rho(M^*).
\end{eqnarray*}
\quad Thus $\eta(M)={\rm ht}(I_\Delta)=\dim(C^\perp)=\rho(M^*)$. The independence 
complex $\Delta$ is pure shellable 
\cite[Remark~1, p.~78]{JohVer}. Hence $I_\Delta$ is Cohen--Macaulay,
that is, ${\rm ht}(I_\Delta)$ is equal to ${\rm pd}_R(R/I_\Delta)$, the
projective dimension of $R/I_\Delta$. Let $\beta_{r,j}$ be the
$(r,j)$-th graded Betti number of $R/I_\Delta$, with 
$r=\eta(M)={\rm pd}_R(R/I_\Delta)$. According to
\cite[Theorem~3.4]{Sta2}, the ring $R/I_\Delta$ is level. Therefore,
by making $r=\dim(C^\perp)$ in Corollary~\ref{dec30-18}, we get
\begin{eqnarray*}
{\rm reg}(R/I_\Delta)&=&\max\{j-r\vert\, \beta_{r,j}\neq
0\}=\min\{j-r\vert\, \beta_{r,j}\neq
0\}\\
&=&\min\{j\vert\, \beta_{r,j}\neq 0\}-r=\delta_r(C^\perp)-\dim(C^\perp).
\end{eqnarray*}
\quad Thus the equality of (a) holds. The equality of (b) follows from (a) using duality. Parts
(c) and (d) follow readily from Remark~\ref{dec29-18}.
\end{proof}

\begin{theorem}\label{regularity-ideal}
Let $G_\sigma$ be a signed graph without loops with $s$ vertices, $m$
edges, $c$ connected components,  
$c_0$ balanced components, let $M$ be the
matroid on $E$ of the incidence matrix code $C$ of $G_\sigma$, over a
finite field of characteristic $p$, 
and let  $\Delta$ $($resp. $\Delta^*$$)$ be the
independence complex of $M$ $($resp. $M^*$$)$. The following hold.
$$
{\rm reg}(R/I)=\begin{cases} 
m-s+c_0&\text{if }I=I_{\Delta^*},\,\ p\neq 2
,\\
 m-s+c&\text{if }I=I_{\Delta^*},\,\ p=2\,  
\mbox{ or }\, G_\sigma\mbox{ is balanced},\\
s-c_0&\text{if }I=I_\Delta,\ p\neq 2,   
\mbox{ and any }i\in E\mbox{ is in some circuit of }M
,\\
 s-c&\text{if }I=I_\Delta,\ p=2 
\mbox{ or }G_\sigma\mbox{ is balanced},\mbox{ and }G\textit{ has no
bridges}.
\end{cases}
$$
\end{theorem}

\begin{proof} Let $A$ be the incidence matrix of $G_\sigma$. As
$G_\sigma$ has no loops, all columns of $A$ are non-zero, that is, $C$
is non-degenerate. Hence, the first two formulas follow at once from
Proposition~\ref{rank-incidence-signed}, Lemma~\ref{dec28-18}(d), and
the equality $\dim(C^\perp)=m-\dim(C)$. 

Assume that $p\neq 2$ and suppose any $i\in E$ is in 
some circuit of $M$. Let $H$ be the parity check matrix of $C$ whose
rows correspond to the circuits of $C$ 
(see the discussion below). The matrix $H$ is a generator matrix for
$C^\perp$ and $M^*$ is the vector matroid $M[H]$. Let $\mathbf{v}_1,\ldots,\mathbf{v}_m$ be
the column vectors of $A$.  Take any $i\in E$, then $i$ is in some
circuit $X\subseteq E$ of $M$. Then $\sum_{j\in
X}\lambda_j\mathbf{v}_j=0$, where $\lambda_j\neq 0$ for $j\in X$. 
Setting $\lambda_j=0$ for $j\in E\setminus X$, we get that 
$\mathbf{\lambda}=(\lambda_1,\ldots,\lambda_m)$ is  
a row of $H$ and $\lambda_i\neq 0$. Thus the $i$-th column of $H$ is
non-zero for $i=1,\ldots,m$, that is, $C^\perp$ is non-degenerate.
Therefore, the third formula follows from
Proposition~\ref{rank-incidence-signed} and Lemma~\ref{dec28-18}(c). 

Assume that $p=2$ or $G_\sigma$ is balanced, and suppose $G$ has no
bridges. Then the vector matroid $M$ is the graphic matroid of $G$. As
$G$ has no bridges, i.e., any edge belongs to a cycle, one has that
every edge is in some circuit of $M$. Hence, by
the previous part, $C^ \perp$ is non-degenerate. Hence, by 
Proposition~\ref{rank-incidence-signed} and Lemma~\ref{dec28-18}(c), the
fourth equality follows. 
\end{proof}

\section{An algebraic formula for the frustration index}\label{frustration-index-section}

Let $G_\sigma$ be a connected signed simple graph with $s$ vertices,
$m$ edges, frustration index $\varphi(G_\sigma)$, and let
$V(G_\sigma)=\{t_1,\ldots,t_s\}$ be its vertex set. For use below,
$\mathbb{X}$ will denote the set of projective points in the
projective space $\mathbb{P}^{s-1}$ defined by the column vectors of the
incidence matrix of $G_\sigma$ over a field $K$ 
of ${\rm char}(K)\neq 2$. Consider a polynomial ring 
$S=K[t_1,\ldots,t_s]=\bigoplus_{d=0}^\infty S_d$ over a field $K$ with the
standard grading. Given a
homogeneous polynomial 
$h$ in $S$, that is, $h\in S_d$ for some $d$, we denote the set of zeros of $h$ in $\mathbb{X}$ by
$V_\mathbb{X}(h)$. The {\it vanishing ideal\/} of
$\mathbb{X}$, denoted $I(\mathbb{X})$,  is the ideal of $S$ 
generated by the homogeneous polynomials that vanish at all points of
$\mathbb{X}$. 

The following characterization of balanced signed graphs is due to 
Harary \cite{Harary}. For other characterizations of this property 
see \cite{Zaslavsky-signed-graphs} and the references therein.

\begin{theorem}{\rm(\cite[Theorem~3]{Harary},
\cite[Proposition~2.1]{Zaslavsky-signed-graphs})}
\label{Harary-balanced} A signed simple graph is balanced if and only if
its vertex set can be partitioned into two disjoint classes $($possible
empty$)$, such that an edge is negative if and only if its two 
endpoints belong to distinct classes. 
\end{theorem}

\begin{lemma}\label{vila-zaslavsky-lemma}
Let $G_\sigma$ be a connected signed simple graph over a
field $K$ of ${\rm char}(K)\neq 2$. Then
\begin{equation} \label{feb20-19}
\varphi(G_\sigma)=\min\{|\mathbb{X}\setminus V_\mathbb{X}(h)|\colon
\, h=a_1t_1+\cdots+a_st_s,\, a_i\in\{\pm 1\}\mbox{ for all }i\}.
\end{equation}
\end{lemma}

\begin{proof} Let $\mathbf{v}_1,\ldots,\mathbf{v}_m$ be
the column vectors of the incidence matrix of $G_\sigma$. 
We set $r=\varphi(G_\sigma)$ and let $r_0$ be
the right hand side of Eq.~(\ref{feb20-19}). 
If $G_\sigma$ is balanced, using Theorem~\ref{Harary-balanced}, it is
not hard to see that there is a linear polynomial
$h=a_1t_1+\cdots+a_st_s$, $a_i\in\{\pm 1\}$ for all $i$, such that
$h(\mathbf{v}_i)=0$ for all $i$, 
that is, $r_0=0$ and $\varphi(G_\sigma)=r_0$ (see the discussion below). Thus we may assume that
$G_\sigma$ is not balanced. Pick a minimum set of edges
$e_1,\ldots,e_r$ such that the signed subgraph 
$H_\sigma=G_\sigma\setminus\{e_1,\ldots,e_r\}$ is balanced. We may
assume that $\{e_1,\ldots,e_r,\ldots,e_m\}$ is the set of edges of
$G_\sigma$ and that $e_i$ corresponds to $\mathbf{v}_i$ for
$i=1,\ldots,m$. We first show the inequality $r\geq r_0$. Note that
$V(H_\sigma)=V(G_\sigma)$. According
to Theorem~\ref{Harary-balanced}, the vertex set of $H_\sigma$ can be
partitioned into two disjoint classes $V_1$ and $V_2$ 
(possible empty) is such a way that an edge of $H_\sigma$ is negative
if and only if its two endpoints belong to distinct classes. We set 
\begin{equation*}
h:=\sum_{t_i\in V_1}t_i-\sum_{t_i\in V_2}t_i.
\end{equation*}
\quad To show the inequality $r\geq r_0$ it suffices to show 
the equality $V_\mathbb{X}(h)=\{\mathbf{v}_i\}_{i=r+1}^m$ because 
this equality implies $r=|\mathbb{X}\setminus V_\mathbb{X}(h)|$, and
consequently $r=\varphi(G_\sigma)\geq r_0$.

\quad Case (I): $V_2=\emptyset$. Therefore, $\sigma(e)=+$ for $e\in
E(H_\sigma)$. As $h=\sum_{i=1}^st_i$, one has the inclusion 
$\{\mathbf{v}_i\}_{i=r+1}^m\subseteq V_\mathbb{X}(h)$. We claim that 
$\sigma(e_i)=-$ for $i=1,\ldots,r$. If
$\sigma(e_i)=+$ for some $1\leq i\leq r$, then 
$G_\sigma\setminus\{e_1,\ldots,e_{i-1},e_{i+1},\ldots,e_r\}$ is balanced 
because it is a positive signed graph, a contradiction. As
$h=\sum_{i=1}^st_i$, the inclusion  
$V_\mathbb{X}(h)\subseteq \{\mathbf{v}_i\}_{i=r+1}^m$ follows because
${\rm char}(K)\neq 2$.  

Case (II): $V_1\neq \emptyset$ and $V_2\neq \emptyset$. If $1\leq
i\leq r$ and $e_i$ joins
$V_1$ and $V_2$, then $\sigma(e_i)=+$ and
$h(\mathbf{v}_i)\neq 0$ because ${\rm char}(K)\neq 2$. Indeed, if
$\sigma(e_i)=-$, then 
$G_\sigma\setminus\{e_1,\ldots,e_{i-1},e_{i+1},\ldots,e_r\}$ is
balanced by Theorem~\ref{Harary-balanced}, a contradiction. If $1\leq
i\leq r$ and the two endpoints of $e_i$ are both in $V_1$ or $V_2$,
then $\sigma(e_i)=-$ and
$h(\mathbf{v}_i)\neq 0$ because ${\rm char}(K)\neq 2$. Indeed, if
$\sigma(e_i)=+$, then 
$G_\sigma\setminus\{e_1,\ldots,e_{i-1},e_{i+1},\ldots,e_r\}$ is
balanced by Theorem~\ref{Harary-balanced}, a contradiction. Thus, one
has the inclusion 
$V_\mathbb{X}(h)\subseteq \{\mathbf{v}_i\}_{i=r+1}^m$. If $i>r$, 
then $h(\mathbf{v}_i)=0$, that is, $\{\mathbf{v}_i\}_{i=r+1}^m\subseteq
V_\mathbb{X}(h)$. This follows noticing that, for $i>r$, one has 
$\sigma(e_i)=+$ if the endpoints of $e_i$ are in $V_1$ or $V_2$, and $\sigma(e_i)=-$
if $e_i$ joins $V_1$ and $V_2$. Therefore the equality
$V_\mathbb{X}(h)=\{\mathbf{v}_i\}_{i=r+1}^m$ 
holds. 

Now, we show the inequality $r\leq r_0$. Pick
$h=a_1t_1+\cdots+a_st_s$, $a_i=\pm 1$ for $i=1,\ldots,s$, such that
$r_0=|\mathbb{X}\setminus V_\mathbb{X}(h)|$. We may assume that the
set $\mathbb{X}\setminus V_\mathbb{X}(h)$ is equal to
$\{\mathbf{v}_1,\ldots,\mathbf{v}_{r_0}\}$, and we may also 
assume that $\{e_1,\ldots,e_{r_0},\ldots,e_m\}$ is the set of edges of
$G_\sigma$ and that $e_i$ corresponds to $\mathbf{v}_i$ for
$i=1,\ldots,m$. It suffices to show that the signed 
subgraph $H_\sigma=G_\sigma\setminus\{e_1,\ldots,e_{r_0}\}$ is
balanced because this implies that $r=\varphi(G_\sigma)\leq r_0$.
There are disjoint sets $V_1$ and $V_2$ (possibly empty) such that 
$V(G_\sigma)=\{t_1,\ldots,t_s\}=V_1\cup V_2$ and 
\begin{equation*}
h=\sum_{t_i\in V_1}t_i-\sum_{t_i\in V_2}t_i.
\end{equation*}
\quad Note that $h(\mathbf{v}_i)=0$ if and only if $i>r_0$ and
$E(H_\sigma)=\{e_{i}\}_{i=r_0+1}^m$. If $\sigma(e_i)=-$ for some
$i>r_0$, then $h(\mathbf{v}_i)=0$, and consequently $e_i$ joins $V_1$
and $V_2$ because ${\rm char}(K)\neq 2$. If $\sigma(e_i)=+$ for some
$i>r_0$, then $h(\mathbf{v}_i)=0$, and consequently the endpoints of
$e_i$ are in $V_1$ or $V_2$. Therefore, by
Theorem~\ref{Harary-balanced},
$H_\sigma=G_\sigma\setminus\{e_1,\ldots,e_{r_0}\}$ is balanced. 
\end{proof}

Let $I\neq(0)$ be a graded ideal
of $S$ of Krull dimension $k$. The {\it Hilbert function} of $S/I$ is: 
$$
H_I(d):=\dim_K(S_d/I_d),\ \ \ d=0,1,2,\ldots,
$$
where $I_d=I\cap S_d$. By a theorem of Hilbert \cite[p.~58]{Sta1},
there is a unique polynomial 
$h_I(x)\in\mathbb{Q}[x]$ of 
degree $k-1$ such that $H_I(d)=h_I(d)$ for  $d\gg 0$. The
degree of the zero polynomial is $-1$. 

The {\it degree\/} or {\it multiplicity\/} of $S/I$, denoted
$\deg(S/I)$, is the positive integer given by 
$$
\deg(S/I):=(k-1)!\, \lim_{d\rightarrow\infty}{H_I(d)}/{d^{k-1}}
\ \mbox{ if }\ k\geq 1,
$$ 
and $\deg(S/I)=\dim_K(S/I)$ if $k=0$. If $f\in S$, the ideal 
$(I\colon f)=\{g\in S\, \vert\, gf \in I \}$ is referred to as
a \textit{colon ideal}. Note that $f$ is a zero-divisor of $S/I$ if and only 
if $(I\colon f)\neq I$.

\begin{lemma}{\cite[Lemma~3.2]{hilbert-min-dis}}
\label{degree-formula-for-the-number-of-zeros-proj}
Let $\mathbb{X}$ be a finite subset of 
$\mathbb{P}^{s-1}$ over a field $K$ and let $I(\mathbb{X})\subset S$
be its vanishing ideal. If 
$0\neq f\in S$ is homogeneous and $(I(\mathbb{X})\colon f)\neq
I(\mathbb{X})$, then  
$$
|V_{\mathbb{X}}(f)|=\deg(S/(I(\mathbb{X}),f)).
$$
\end{lemma}

The following algebraic formula for the frustration index can be used
to compute or estimate this number using \textit{Macaulay}$2$
\cite{mac2} (Example~\ref{example-graph4}). 

\begin{theorem}\label{vila-zaslavsky}
Let $G_\sigma$ be a connected unbalanced signed simple graph with
frustration index $\varphi(G_\sigma)$ over a
field $K$ of ${\rm char}(K)\neq 2$, and let $\mathcal{F}$ be the set
of linear forms $h=\sum_{i=1}^sa_it_i$ such that $a_i=\pm 1$ for all
$i$ and $(I(\mathbb{X})\colon h)\neq I(\mathbb{X})$. Then
\begin{equation*}
\varphi(G_\sigma)=
|\mathbb{X}|-\max\{\deg(S/(I(\mathbb{X}),h))\colon
\, h\in\mathcal{F}\}.
\end{equation*}
\end{theorem}

\begin{proof} The vanishing ideal $I(\mathbb{X})$ does not contains
linear forms. This follows by noticing that the incidence matrix of
$G_\sigma$ has rank equal to $s$, the number of vertices of $G_\sigma$,
because $G_\sigma$ is unbalanced and connected (see
Proposition~\ref{rank-incidence-signed}). Thus $\mathbb{X}\setminus
V_\mathbb{X}(h)\neq \emptyset$ for any $0\neq h\in S_1$. If $h$ is a
linear form, by \cite[Lemma~3.1]{rth-footprint},
$V_\mathbb{X}(h)\neq\emptyset$ if and only if $(I(\mathbb{X})\colon
h)\neq I(\mathbb{X})$. Therefore, using
Lemmas~\ref{vila-zaslavsky-lemma} and
\ref{degree-formula-for-the-number-of-zeros-proj}, we obtain 
\begin{eqnarray*}
\varphi(G_\sigma)&=&\min\{|\mathbb{X}\setminus V_\mathbb{X}(h)|\colon
\, h=\textstyle\sum_{i=1}^sa_it_i,\, a_i=\pm 1\mbox{ for all }i\}\\
&=&\min\{|\mathbb{X}\setminus V_\mathbb{X}(h)|\colon
\, h=\textstyle\sum_{i=1}^sa_it_i,\, a_i=\pm 1\mbox{ for all }i\mbox{ and
}V_\mathbb{X}(h)\neq\emptyset\}\\
&=& \min\{|\mathbb{X}\setminus V_\mathbb{X}(h)|\colon
\, h\in\mathcal{F}\}=|\mathbb{X}|-\max\{|V_\mathbb{X}(h)|\colon
\, h\in\mathcal{F}\}\\
&=&|\mathbb{X}|-\max\{\deg(S/(I(\mathbb{X}),h))\colon
\, h\in\mathcal{F}\}.
\end{eqnarray*}
\quad The second equality follows by discarding all $h$ with
$V_\mathbb{X}(h)=\emptyset$.
\end{proof}

\begin{remark}\label{dec8-18} If we allow the coefficients
$a_1,\ldots,a_s$ to be in $\{0,\pm 1\}$
such that not all of them are zero, we obtain the minimum distance of
the incidence matrix code $C$ of $G_\sigma$ over any finite field of
characteristic $p\neq 2$. This follows from the results of
Section~\ref{GHW-section} and Proposition~\ref{mar10-19} below. 
\end{remark}

The following algebraic formula for the minimum distance of an incidence matrix
code can be used
to compute or estimate this number using \textit{Macaulay}$2$
\cite{mac2} and the algorithms of \cite{rth-footprint,hilbert-min-dis}. 

\begin{proposition}\label{mar10-19} 
Let $G_\sigma$ be a connected signed simple graph and let $C$ be its
incidence matrix code over a finite field $K$.  
Then the minimum distance of $C$ is given by 
\begin{equation*}
\delta(C)=
|\mathbb{X}|-\max\{\deg(S/(I(\mathbb{X}),h))\colon
\, h\in S_1\setminus I(\mathbb{X})\mbox{ and } (I(\mathbb{X})\colon h)\neq I(\mathbb{X})\}.
\end{equation*}
\end{proposition}

\begin{proof} Let $\mathbf{v}_1,\ldots,\mathbf{v}_m$ be
the column vectors of the incidence matrix of $G_\sigma$ and let
$P_i$ be the point $[\mathbf{v}_i]$ in
$\mathbb{P}^{s-1}$ for $i=1,\ldots,m$. Thus, $\mathbb{X}$ is
the set of points $\{P_1,\ldots,P_m\}$. Note that $C$ is the image 
of $S_1$---the vector space of linear
forms of $S$---under the evaluation map 
$${\rm ev}_1\colon S_1\rightarrow K^{m},\quad
h\mapsto\left(h(P_1),\ldots,h(P_m)\right).
$$ 
\quad The image of the linear function $t_i$, under the
map $\text{ev}_1$, gives the $i$-th row of $C$. This means that $C$ is the Reed--Muller-type code
$C_\mathbb{X}(1)$ in the sense of \cite{GRT}. The result now follows
readily by applying \cite[Theorem~4.7]{hilbert-min-dis}.
\end{proof}
\section{Examples of signed graphs}\label{examples-section}
In this section we illustrate how to use our results in practice with
some examples. 

\begin{example} Let $G_\sigma$ be a signed simple graph whose underlying
graph $G$ is given in  Figure~\ref{complete-intersection}, let $C$ be
the incidence matrix code of $G_\sigma$, let $A$
be the incidence matrix of $G_\sigma$, and let $M=M[A]$ be the 
matroid of $C$. Assume that $K$ is either a
field of characteristic $2$ or that $K$ is any field and 
$G_\sigma=G_+$. In either case, by Theorem~\ref{signed-graphic-rep}(b) and
Corollary~\ref{circuits-signed-graphic}(d), $M$ is the cycle matroid
of $G$ and, by
Proposition~\ref{rank-incidence-signed}, the rank of $M$ is $10$. 
\begin{figure}[ht]
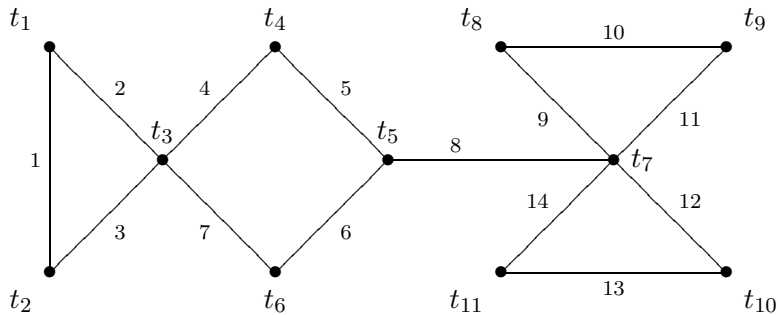

\[
 \xygraph{
 !{<0cm,0cm>:<1.5cm,0cm>:<0cm,1.5cm>::}
 !{(0,2)}*+{\bullet}@\cir{}="v1"
 !{(1,1)}*+{\bullet}@\cir{}="v2"
 !{(0,0)}*+{\bullet}@\cir{}="v3"
 !{(2,2)}*+{\bullet}@\cir{}="v4"
 !{(3,1)}*+{\bullet}@\cir{}="v5"
 !{(2,0)}*+{\bullet}@\cir{}="v6"
 !{(5,1)}*+{\bullet}@\cir{}="v7"
 !{(4,2)}*+{\bullet}@\cir{}="v8"
 !{(6,2)}*+{\bullet}@\cir{}="v9"
 !{(6,0)}*+{\bullet}@\cir{}="v10"
 !{(4,0)}*+{\bullet}@\cir{}="v11"
 "v1"-_{1}"v3" "v1"-^{2}"v2" "v2"-^{3}"v3"
 "v2"-^{4}"v4" "v4"-^{5}"v5" "v5"-^{6}"v6"
 "v6"-^{7}"v2" "v5"-^{\hspace{-1.2cm}8}"v7" "v7"-^{9}"v8"
 "v8"-^{10}"v9" "v9"-^{11}"v7" "v7"-^{12}"v10"
 "v10"-^{13}"v11" "v11"-^{14}"v7"
 "v1"!{+U*++!RD{t_1}} "v2"!{+D*++!D{t_3}} "v3"!{+U*++!RU{t_2}}
 "v4"!{+U*++!D{t_4}} "v5"!{+U*++!D{t_5}} "v6"!{+U*++!U{t_6}}
 "v7"!{+U*++!L{t_7}} "v8"!{+U*++!RD{t_8}} "v9"!{+U*++!LD{t_9}}
 "v10"!{+U*++!LU{t_{10}}} "v11"!{+U*++!RU{t_{11}}}
 }
\]
\caption{Simple graph $G$ with $11$ vertices and $14$ edges}\label{complete-intersection}
\end{figure}

Therefore, the circuits of $M$ are the cycles of $G$ and they are given by
\begin{center}
\begin{tabular}{l}
 $c_1=\{1,2,3\}$,\quad $c_2=\{9,10,11\}$,\quad
 $c_3=\{12,13,14\}$,\quad  $c_4=\{4,5,6,7\}$.
 \end{tabular}
\end{center}
\quad Hence, by applying Theorem~\ref{pepe-vila-2019-dual}(a), we get the
generalized Hamming weights of $C^\perp$:
 \begin{center}{\rm
 \begin{tabular}{|c|c|c|c|c|}\hline
 $r$&1&2&3&4\\\hline
 $\delta_r(C^\perp)$&3&6&9&13\\\hline
 \end{tabular}}\,.
 \end{center}
\quad Concretely, one has $\delta_r(C^\perp)=|c_1\cup\cdots\cup c_r|$
for $1\leq r\leq 4$. Let $R=K[x_1,\ldots,x_{14}]$ be a polynomial ring over the
field $K$. The ideal of circuits of $M$ is the squarefree
monomial ideal $I=I(\mathcal{C}_M)$ of $R$ generated
by all monomials $\prod_{j\in c_i}x_j$ with $i=1,\ldots,4$. Using
\textit{Macaulay}$2$ \cite{mac2}, we obtain that the minimal free
resolution of $R/I$ is:
\[0 \rightarrow R(-13) \rightarrow R(-9)\oplus R^3(-10) \rightarrow
R^3(-6)\oplus R^3(-7)
\rightarrow R^3(-3)\oplus R(-4)
\rightarrow R\rightarrow R/I\rightarrow 0\] 
\quad One can verify the values of the $\delta_r(C^\perp)$'s by applying 
Corollary~\ref{dec30-18} to this resolution. By Wei's duality
(Theorem~\ref{wei-duality}), one has
 \begin{center}{\rm
 \begin{tabular}{|c|c|c|c|c|c|c|c|c|c|c|}\hline
 $r$&1&2&3&4&5&6&7&8&9&10\\\hline
 $\delta_r(C)$&1&3&4&5&7&8&10&11&13&14\\\hline
 \end{tabular}}\,.
 \end{center}
\quad According to Theorem~\ref{pepe-vila-2019}, 
$\delta_r(C)=\lambda_r(C)$ for $r=1,\ldots,10$. Removing edge $8$ from
$G$, we get two connected
components. Thus $\delta_1(C)=1$. To illustrate the equality 
$\delta_7(C)=10$, note that removing 
the ten edges that are not in the square of the 
graph $G$ results in a subgraph with eight connected
components, and $\lambda_7(C)=10$. The edge biparticity of $G$ is
$\varphi(G_-)=3$. 
\end{example}

\begin{example} Let $G$ be the graph of
Figure~\ref{complete-intersection}, let $K$ be a field 
of ${\rm char}(K)\neq 2$, and let $C$ be the incidence matrix code of
$G_-$. By Corollary~\ref{circuits-signed-graphic}(c), the circuits of the negative signed
graph $G_-$, that is, the circuits of the signed-graphic matroid
$M(G_-)$, are the even cycles and the bowties of $G$:
 \begin{center}{\rm
 \begin{tabular}{l}
$c_1=\{4, 5, 6, 7\}$,\ $c_2=\{9, 10, 11, 12, 13, 14\}$,\\ $c_3=\{1,
2, 3, 4, 5,
8, 9, 10, 11\}$,\ $c_4=\{1, 2,
    3, 4, 5, 8, 12, 13, 14\}$,\\ $c_5=\{1, 2, 3, 6, 7, 8, 9, 10,
    11\}$,\ 
    $c_6=\{1, 2, 3, 6, 7, 8, 12, 13, 14\}$.\end{tabular}}
\end{center}
\quad Hence, by Theorem~\ref{pepe-vila-2019-dual}(b), it follows that
$\delta_r(C^\perp)=|c_1\cup\cdots\cup c_r|$ for $1\leq
r\leq 3$, and we obtain the generalized Hamming weights of $C^\perp$:
 \begin{center}{\rm
 \begin{tabular}{|c|c|c|c|}\hline
 $r$&1&2&3\\\hline
 $\delta_r(C^\perp)$&4&10&14\\\hline
 \end{tabular}}\,.
 \end{center}
\quad Let $R=K[x_1,\ldots,x_{14}]$ be a polynomial ring over the
field $K$ and let $I=I(\mathcal{C}_M)\subset R$ be the ideal of
circuits of the signed-graphic matroid $M(G_-)$. Using
\textit{Macaulay}$2$ \cite{mac2}, we obtain that the minimal free
resolution of $R/I$ is:
\[0 \rightarrow R^4(-14) \rightarrow R(-10)\oplus R^4(-11)\oplus
R^4(-12) 
\rightarrow R(-4)\oplus R(-6)\oplus R^4(-9)
\rightarrow  R\rightarrow R/I\rightarrow 0.\] 
\quad One can verify the values of the $\delta_r(C^\perp)$'s by applying 
Corollary~\ref{dec30-18} to this resolution. By Wei's duality (Theorem~\ref{wei-duality}), 
we obtain the generalized Hamming weights of $C$:
 \begin{center}{\rm
 \begin{tabular}{|c|c|c|c|c|c|c|c|c|c|c|c|}\hline
 $r$&1&2&3&4&5&6&7&8&9&10&11\\\hline
 $\delta_r(C)$&2&3&4&6&7&8&9&10&12&13&14\\\hline
 \end{tabular}}\,.
 \end{center} 
\quad According to Theorem~\ref{pepe-vila-2019},  
$\delta_r(C)=\upsilon_r(G_-)$ for $r=1,\ldots,11$. Next we verify these
values. Removing edges $2$
and $3$ from 
$G$, we get a graph with a bipartite component. 
Therefore, by Theorem~\ref{pepe-vila-2019}, $\delta_1(C)=2$. 
To check the other values of $\delta_r(C)$ using by
Theorem~\ref{pepe-vila-2019},  
 note that successively removing from 
the graph $G$ the edges
$$\{1,2\},\, 3,\, 8,\, \{10,13\},\, 9,\, 11,\, 12,\, 14,\, \{4,5\},\,
6,\, 7,$$ 
we obtain a subgraph
with $r$ bipartite connected
components at the $r$-th step. By Theorem~\ref{regularity-ideal}, the
regularity of $R/I$ is $11$. The frustration index of $G_-$ is $3$
which is the edge biparticity of $G$.
\end{example}

\begin{example}\label{example-graph1} Let $G_\sigma$ be the signed graph of
Figure~\ref{signed-graph-example}, let $C$ be the incidence matrix
code of $G_\sigma$ over a finite field of ${\rm char}(K)=p\neq 2$, and 
let $M=M[A]$ be the vector matroid of $C$, where $A$
is the incidence matrix of $G_\sigma$.
\begin{figure}[ht]
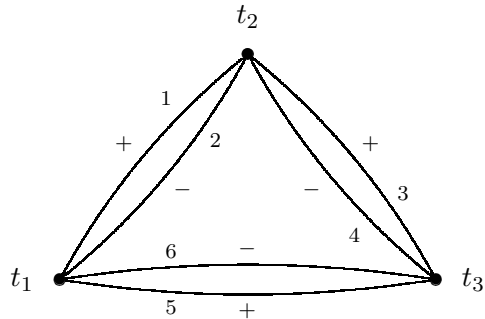
 
$$
 \xy
 \POS (-25,0)*{\bullet}*\cir<2pt>{}="a"
 \POS (25,0)*{\bullet}*\cir<2pt>{}="b"
 \POS (0,30)*{\bullet}*\cir<2pt>{}="c"
 \POS "a"\ar@{-}@/^1.2ex/"b"^(0.3)*{_6}^*{_-}
 \POS "b"\ar@{-}@/^1.2ex/"a"^*{_+}^(0.7)*{_5}
 \POS "b"\ar@{-}@/^1.2ex/"c"^(0.3)*{_4}^*{_-}
 \POS "c"\ar@{-}@/^1.2ex/"b"^*{_+}^(0.7)*{_3}
 \POS "c"\ar@{-}@/^1.2ex/"a"^(0.3)*{_2}^*{_-}
 \POS "a"\ar@{-}@/^1.2ex/"c"^*{_+}^(0.7)*{_1}
 \POS (-30,0)*{t_1}
 \POS (30,0)*{t_3}
 \POS (0,35)*{t_2}
 \endxy
$$
\caption{Signed graph with $3$ vertices and $6$ edges}
\label{signed-graph-example}
\end{figure}

\quad The incidence matrix of the signed graph $G_\sigma$ is
$$
A=\left[
\begin{matrix}
1&1&0&0&1&1\\
-1&1&1&1&0&0\\
0&0&-1&1&-1&1
\end{matrix}
\right]\!.
$$
\quad Using Procedure~\ref{procedure-example-graph1}, we obtain the
following information. The ideals of circuits and cocircuits of $M$ are given by 
\begin{eqnarray*}
I&=&(x_1x_2x_3x_4,\, x_1x_3x_5,\, x_2x_4x_5,\, x_2x_3x_6,\, x_1x_4x_6,\,
x_1x_2x_5x_6,\, x_3x_4x_5x_6),\\
I^*&=&(x_2x_4x_6,\,  x_1x_3x_6,\,  x_1x_4x_5,\,  x_2x_3x_5,\,
x_3x_4x_5x_6
,\, x_1x_2x_5x_6,\,  x_1x_2x_3x_4),
\end{eqnarray*}
and ${\rm reg}(R/I)={\rm reg}(R/I^*)=3$. The generalized Hamming
weights of $C^\perp$ and $C$ are
\begin{eqnarray*}
 \begin{tabular}{|c|c|c|c|}\hline
 $r$&1&2&3\\\hline
 $\delta_r(C^\perp)$&3&5&6\\\hline
 \end{tabular} &\ \ \ \ \ \ &
 \begin{tabular}{|c|c|c|c|}\hline
 $r$&1&2&3\\\hline
 $\delta_r(C)$&3&5&6\\\hline
 \end{tabular}\,.
\end{eqnarray*}
\quad Thus, by Theorem~\ref{pepe-vila-2019}, the cogirth of the signed graph $G_\sigma$
is $\upsilon_1(G_\sigma)=3$, and one has $\upsilon_2(G_\sigma)=5$,
$\upsilon_3(G_\sigma)=6$. The frustration index of $G_\sigma$ is $3$.
\end{example}

\begin{example}\label{example-graph2} Let $G_+$ be the positive signed
graph of Figure~\ref{positive-signed-graph-example}, let $C$ be the incidence matrix
code of $G_+$ over a finite field $K$, and 
let $M=M[A]$ be the vector matroid of $C$, where $A$ 
is the incidence matrix of $G_+$. By
Corollary~\ref{circuits-signed-graphic}(b), $M$ is the graphic matroid of
the underlying graph $G$, that is, the circuits and cocircuits of $M$
are the cycles and cocycles of $G$. 
\begin{figure}[ht]
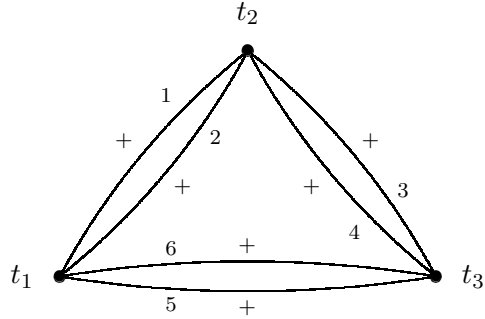
 
$$
 \xy
 \POS (-25,0)*{\bullet}*\cir<2pt>{}="a"
 \POS (25,0)*{\bullet}*\cir<2pt>{}="b"
 \POS (0,30)*{\bullet}*\cir<2pt>{}="c"
 \POS "a"\ar@{-}@/^1.2ex/"b"^(0.3)*{_6}^*{_+}
 \POS "b"\ar@{-}@/^1.2ex/"a"^*{_+}^(0.7)*{_5}
 \POS "b"\ar@{-}@/^1.2ex/"c"^(0.3)*{_4}^*{_+}
 \POS "c"\ar@{-}@/^1.2ex/"b"^*{_+}^(0.7)*{_3}
 \POS "c"\ar@{-}@/^1.2ex/"a"^(0.3)*{_2}^*{_+}
 \POS "a"\ar@{-}@/^1.2ex/"c"^*{_+}^(0.7)*{_1}
 \POS (-30,0)*{t_1}
 \POS (30,0)*{t_3}
 \POS (0,35)*{t_2}
 \endxy
$$
\caption{Positive signed graph with $3$ vertices and $6$ edges}
\label{positive-signed-graph-example}
\end{figure}

\quad The incidence matrix of the positive signed graph $G_+$ is
$$
A=\left[
\begin{array}{rrrrrr}
1 & 1&0&0&1&1\\
-1&-1&1&1&0&0\\
0&0& -1&-1&-1&-1
\end{array}
\right]\!.
$$
\quad Using Procedure~\ref{procedure-example-graph2}, we obtain the
following information. The ideals of circuits and cocircuits of $M$ are given by 
\begin{eqnarray*}
I&=&
(x_5x_6,\,  x_3x_4,\,  x_1x_2,\,  x_2x_4x_6,\,  x_1x_4x_6,\,  x_2x_3x_6,\, 
x_1x_3x_6,\, x_2x_4x_5,\,  x_1x_4x_5,\,  x_2x_3x_5,\,  x_1x_3x_5),\\
I^*&=&(x_3x_4x_5x_6,\, x_1x_2x_5x_6,\, x_1x_2x_3x_4),
\end{eqnarray*}
${\rm reg}(R/I)=2$, and ${\rm reg}(R/I^*)=4$. The generalized Hamming
weights of $C^\perp$ and $C$ are
\begin{eqnarray*}
 \begin{tabular}{|c|c|c|c|c|}\hline
 $r$&1&2&3&4\\\hline
 $\delta_r(C^\perp)$&2&4&5&6\\\hline
 \end{tabular} &\ \ \ \ \ \ &
 \begin{tabular}{|c|c|c|}\hline
 $r$&1&2\\\hline
 $\delta_r(C)$&4&6\\\hline
 \end{tabular}\,.
\end{eqnarray*}
\quad Thus, by Theorem~\ref{pepe-vila-2019}, the edge connectivity 
of $G$ is $\lambda_1(G)=4$, and $\lambda_2(G)=6$.
\end{example}

\begin{example}\label{example-graph3} Let $G_-$ be the negative signed
graph of Figure~\ref{negative-signed-graph-example}, let $C$ be the incidence matrix
code of $G_-$ over a field $K$ of characteristic $p\neq 2$, and 
let $M=M[A]$ be the vector matroid of $C$, where $A$ 
is the incidence matrix of $G_-$. By Corollary~\ref{VectorCircuit}, $M$ is the even cycle matroid of
the underlying graph $G$, that is, the circuits of $M$
are the even cycles and bowties of $G$.
\begin{figure}[ht]
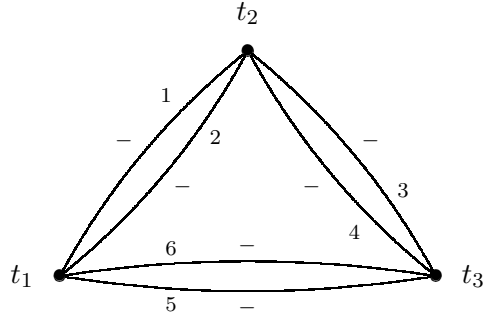
 
$$
 \xy
 \POS (-25,0)*{\bullet}*\cir<2pt>{}="a"
 \POS (25,0)*{\bullet}*\cir<2pt>{}="b"
 \POS (0,30)*{\bullet}*\cir<2pt>{}="c"
 \POS "a"\ar@{-}@/^1.2ex/"b"^(0.3)*{_6}^*{_-}
 \POS "b"\ar@{-}@/^1.2ex/"a"^*{_-}^(0.7)*{_5}
 \POS "b"\ar@{-}@/^1.2ex/"c"^(0.3)*{_4}^*{_-}
 \POS "c"\ar@{-}@/^1.2ex/"b"^*{_-}^(0.7)*{_3}
 \POS "c"\ar@{-}@/^1.2ex/"a"^(0.3)*{_2}^*{_-}
 \POS "a"\ar@{-}@/^1.2ex/"c"^*{_-}^(0.7)*{_1}
 \POS (-30,0)*{t_1}
 \POS (30,0)*{t_3}
 \POS (0,35)*{t_2}
 \endxy
$$
\caption{Negative signed graph with $3$ vertices and $6$ edges}
\label{negative-signed-graph-example}
\end{figure}

\quad The incidence matrix of the negative signed graph $G_-$ is 
the incidence matrix of $G$:
$$
A=\left[
\begin{array}{rrrrrr}
1 & 1&0&0&1&1\\
1&1&1&1&0&0\\
0&0& 1&1&1&1
\end{array}
\right]\!.
$$
\quad Using Procedure~\ref{procedure-example-graph3}, we obtain the
following information. The ideals of circuits and cocircuits of $M$ are given by 
$$
I=
(x_1x_2,\,  x_3x_4,\,  x_5x_6),\ \ I^*=(x_1x_2,\,  x_3x_4,\,  x_5x_6),
$$
${\rm reg}(R/I)={\rm reg}(R/I^*)=3$. The generalized Hamming
weights of $C^\perp$ and $C$ are
\begin{eqnarray*}
 \begin{tabular}{|c|c|c|c|c|}\hline
 $r$&1&2&3\\\hline
 $\delta_r(C^\perp)$&2&4&6\\\hline
 \end{tabular} &\ \ \ \ \ \ &
 \begin{tabular}{|c|c|c|c|}\hline
 $r$&1&2&3\\\hline
 $\delta_r(C)$&2&4&6\\\hline
 \end{tabular}\,.
\end{eqnarray*}
\quad Thus, by Theorem~\ref{pepe-vila-2019}, the cogirth of $G_-$ is 
$\upsilon_1(G_-)=2$, and $\upsilon_2(G_-)=4$, 
$\upsilon_3(G_-)=6$.
\end{example}

\begin{example}\label{example-graph4} Let $G_\sigma$ be the signed graph of
Figure~\ref{signed-graph-example4} and let $G$ be its underlying
graph. 
The incidence matrix of $G_\sigma$ is given in
Procedure~\ref{procedure-example-graph4}. Using this procedure we
obtain that the frustration index $\varphi(G_\sigma)$ of $G_\sigma$ 
is $7$ and the frustration index $\varphi(G_-)$ of the negative signed
graph $G_-$ is $6$. The minimum distance $\delta(C)$ of the incidence matrix code
$C$ of $G_\sigma$ is $4$ if ${\rm char}(K)\neq 2$ and $\delta(C)$ is
$3$ if   ${\rm char}(K)=2$. In this case $\delta(C)=\delta(C^\perp)$ in
any characteristic.
\begin{figure}[ht]
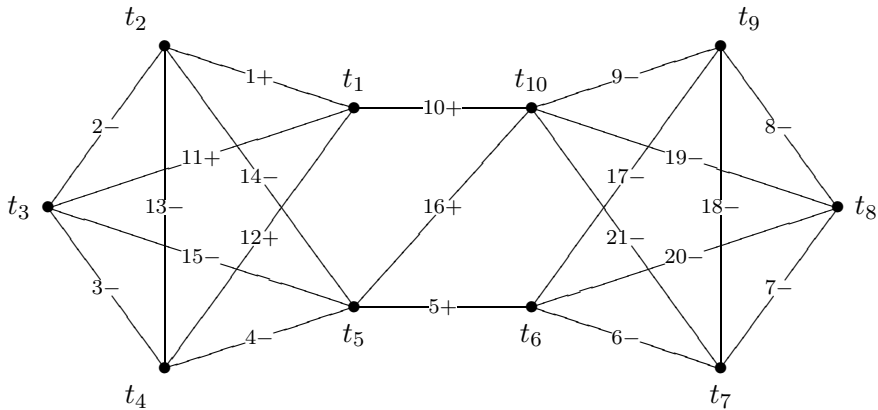
 
$$
 \xygraph{
 !{<0cm,0cm>:<0cm,1.5cm>:<-1.05cm,0cm>}
 !{(0,-4);a(18)**{}?(1.5)}*{\bullet}@\cir{}="a2"
 !{(0,-4);a(90)**{}?(1.5)}*{\bullet}@\cir{}="a3"
 !{(0,-4);a(162)**{}?(1.5)}*{\bullet}@\cir{}="a4"
 !{(0,-4);a(234)**{}?(1.5)}*{\bullet}@\cir{}="a5"
 !{(0,-4);a(306)**{}?(1.5)}*{\bullet}@\cir{}="a1"
 !{(0,4);a(-18)**{}?(1.5)}*{\bullet}@\cir{}="a9"
 !{(0,4);a(54)**{}?(1.5)}*{\bullet}@\cir{}="a10"
 !{(0,4);a(126)**{}?(1.5)}*{\bullet}@\cir{}="a6"
 !{(0,4);a(198)**{}?(1.5)}*{\bullet}@\cir{}="a7"
 !{(0,4);a(270)**{}?(1.5)}*{\bullet}@\cir{}="a8"
 "a1"-|{1+}"a2" "a2"-|{2-}"a3" "a3"-|{3-}"a4" "a4"-|{4-}"a5"
 "a1"-|{11+}"a3" "a1"-|{12+}"a4" "a2"-|{13-}"a4" "a2"-|{14-}"a5" "a3"-|{15-}"a5"
 "a6"-|{6-}"a7" "a7"-|{7-}"a8" "a8"-|{8-}"a9" "a9"-|{9-}"a10"
 "a6"-|{20-}"a8" "a6"-|{17-}"a9" "a7"-|{18-}"a9" "a7"-|{21-}"a10" "a8"-|{19-}"a10"
 "a1"-|{10+}"a10" "a5"-|{5+}"a6" "a5"-|{16+}"a10"
 "a2"!{+U*++!RD{t_2}} "a3"!{+D*++!R{t_3}} "a4"!{+U*++!RU{t_4}}
 "a5"!{+U*++!U{t_5}} "a1"!{+U*++!D{t_1}}
 "a9"!{+U*++!LD{t_9}} "a10"!{+D*++!D{t_{10}}} "a6"!{+U*++!U{t_6}}
 "a7"!{+U*++!U{t_7}} "a8"!{+U*++!L{t_8}}
 }
$$
\caption{Unbalanced signed graph with frustration index $7$}
\label{signed-graph-example4}
\end{figure}
\end{example}

\section*{Acknowledgments} 
We thank Thomas Zaslavsky for suggesting to generalize our work on
incidence matrix codes of graphs to signed graphs, and for pointing
out that the edge biparticity of a graph is a special case of the 
frustration index of a signed graph. 
Computations with \textit{Macaulay}$2$ \cite{mac2}, Matroids
\cite{matroids}, and SageMath \cite{sage} 
were important to verifying and computing examples given in this 
paper. 

\begin{appendix}

\section{Procedures for Macaulay2 and Matroids}\label{procedures-section}

In this section we give 
procedures for \textit{Macaulay}$2$ \cite{mac2}, using the field 
of rational numbers as the ground field, to compute the generalized
Hamming weights of the incidence matrix code of a signed graph and the 
corresponding graph theoretical invariants ($r$-th cogirth, $r$-th
edge connectivity), as well as the ideals of
circuits, cocircuits, cycles and cocycles of a signed graph, and their
algebraic invariants (Betti numbers, shifts, regularity).  We also
give a procedure to compute the frustration index of a connected
signed simple graph. In all procedures
the input is a rational matrix. The package
\textit{Matroids} \cite{matroids} plays an important role here 
because it computes the circuits and cocircuits of a vector matroid 
over the field $\mathbb{Q}$ of rational numbers. 

\begin{procedure}\label{procedure-example-graph1}
Given the incidence matrix $A$ of a signed graph $G_\sigma$ over a field of ${\rm
char}(K)\neq 2$, the procedure below computes the following:
\begin{itemize}
\item The ideal of circuits and the ideal of cocircuits of $G_\sigma$, and its regularity. 
\item The graded Betti numbers of the ideal of circuits and
the ideal of cocircuits of $G_\sigma$.
\item The weight hierarchies of
the incidence matrix code $C$ of $G_\sigma$ and of its dual code $C^\perp$.
\item The $r$-th cogirth of $G_\sigma$ (Theorem~\ref{pepe-vila-2019}). 
\end{itemize}
\quad The next procedure corresponds to Example~\ref{example-graph1}.
To compute other examples just change the incidence matrix $A$. 
\begin{verbatim}
--Procedure for Macaulay2
loadPackage "Matroids"
loadPackage "BoijSoederberg"
A=transpose matrix{{1,-1,0},{1,1,0},{0,1,-1},{0,1,1},{1,0,-1},{1,0,1}}
MA=matroid(A), I=ideal(MA)
m=matrix{flatten entries gens gb I}
N=coker m, F=res N, B=betti F, regularity N
lowestDegrees B --gives the weight hierarchy of the dual of C 
I=ideal(dual(MA))
m=matrix{flatten entries gens gb I}
N=coker m, F=res N, B=betti F, regularity N
lowestDegrees B --gives the weight hierarchy of C 
\end{verbatim}
\end{procedure}

\begin{procedure}\label{procedure-example-graph2}
Using the incidence matrix $A$ of a positive signed graph $G_+$ over a
field $K$ and the Procedure~\ref{procedure-example-graph1}, we can compute the following:
\begin{itemize}
\item The ideal of cycles and the ideal of cocycles of $G$ and its
regularity. 
\item The graded Betti numbers of the ideals of cycles and
cocycles.
\item The weight hierarchies of
the incidence matrix code of $G_+$ and of its dual code, and the
generalized Hamming weights of the incidence matrix code of a
digraph $\mathcal{D}$.
\item The $r$-th edge connectivity of $G$.  
\end{itemize}
\quad The next incidence matrix corresponds to Example~\ref{example-graph2}.
\begin{verbatim}
--Incidence matrix for Macaulay2
A=transpose matrix{{1,-1,0},{1,-1,0},{0,1,-1},{0,1,-1},{1,0,-1},{1,0,-1}}
\end{verbatim}
\end{procedure}

\begin{procedure}\label{procedure-example-graph3}
Using the incidence matrix $A$ of a negative signed graph $G_-$ over a
field $K$ of characteristic $p\neq 2$ and the 
Procedure~\ref{procedure-example-graph1}, we can compute the following:
\begin{itemize}
\item The ideal $I$ of the even cycles and bowties of $G$ and
the ideal $I^*$ of cocircuits of $G_-$.
\item The graded Betti numbers of $I$ and $I^*$, and its regularity. 
\item The weight hierarchies of 
the incidence matrix code of $G_-$ and of its dual code.
\item The $r$-th cogirth of $G_-$.  
\end{itemize}
\quad The next incidence matrix corresponds to Example~\ref{example-graph3}.
\begin{verbatim}
--Incidence matrix for Macaulay2
A=transpose matrix{{1,1,0},{1,1,0},{0,1,1},{0,1,1},{1,0,1},{1,0,1}}
\end{verbatim}
\end{procedure}

\begin{procedure}\label{procedure-example-graph4}
One can use Theorem~\ref{vila-zaslavsky} and \textit{Macaulay}$2$
\cite{mac2} to compute the frustration index of a connected unbalanced
signed simple graph $G_\sigma$. The incidence matrix of the following 
procedure corresponds to the graph of Figure~\ref{signed-graph-example4} given in 
Example~\ref{example-graph4}. 
\begin{verbatim}
--Procedure for Macaulay2
input "points.m2"
R = QQ[t1,t2,t3,t4,t5,t6,t7,t8,t9,t10]
A = transpose matrix{{1,-1,0,0,0,0,0,0,0,0},{0,1,1,0,0,0,0,0,0,0},
{0,0,1,1,0,0,0,0,0,0},{0,0,0,1,1,0,0,0,0,0},{0,0,0,0,1,-1,0,0,0,0},
{0,0,0,0,0,1,1,0,0,0},{0,0,0,0,0,0,1,1,0,0},{0,0,0,0,0,0,0,1,1,0},
{0,0,0,0,0,0,0,0,1,1},{1,0,0,0,0,0,0,0,0,-1},{1,0,-1,0,0,0,0,0,0,0},
{1,0,0,-1,0,0,0,0,0,0},{0,1,0,1,0,0,0,0,0,0},{0,1,0,0,1,0,0,0,0,0},
{0,0,1,0,1,0,0,0,0,0},{0,0,0,0,1,0,0,0,0,-1},{0,0,0,0,0,1,0,0,1,0},
{0,0,0,0,0,0,1,0,1,0},{0,0,0,0,0,0,0,1,0,1},{0,0,0,0,0,1,0,1,0,0},
{0,0,0,0,0,0,1,0,0,1}}
I=ideal(projectivePointsByIntersection(A,R)) 
M=coker gens gb I, G=gb I
frustration=degree M-max apply(apply(subsets(apply(apply(apply
(toList ((set{1}**(set(1,-1))^**(hilbertFunction(1,M)-1))/splice)-
(set{0})^**(hilbertFunction(1,M)),toList),x->basis(1,M)*vector x),
z->ideal(flatten entries z)),1),ideal),x-> if #set flatten entries 
mingens ideal(leadTerm gens x)==1 and not quotient(I,x)==I 
then degree(I+x) else 0)--This gives the frustration index  
\end{verbatim}
\end{procedure}
\end{appendix}

\bibliographystyle{plain}

\end{document}